\documentclass[twocolumn,aps,pra,nofootinbib,showpacs,superscriptaddress]{revtex4-1}

\usepackage{graphicx}
\usepackage{bm}
\usepackage{amsmath,amsthm}
\usepackage[hidelinks]{hyperref}
\usepackage{amssymb}
\usepackage{appendix}
\usepackage{amsfonts}
\usepackage{fancyhdr}
\usepackage{slashed}
\usepackage{latexsym,epsfig,bbm}
\usepackage[none]{hyphenat}
\usepackage{qcircuit}

 \theoremstyle{plain}
\newtheorem{defi}{Definition}
\newtheorem{theorem}{Theorem}

\newtheorem{lemma}{Lemma}
\newtheorem{result}{Result}

\newcommand{\bra}[1]{\langle{#1}|}
\newcommand{\ket}[1]{|{#1}\rangle}

\usepackage{color}
\definecolor{blue}{rgb}{0,0.2,1}

\definecolor{red}{rgb}{0.9,0,0}

\newcommand{\Ord}[1]{\mathcal{O}\left( #1 \right)}
\newcommand{\tOrd}[1]{\tilde{\mathcal{O}}\left( #1 \right)}

\begin{document}

\title{Quantum computational finance: Monte Carlo pricing of financial derivatives}
\author{Patrick Rebentrost}
\email{pr@patrickre.com}
\author{Brajesh Gupt}
\email{brajesh@xanadu.ai}
\author{Thomas R. Bromley}
\email{tom@xanadu.ai}
\affiliation{Xanadu, 372 Richmond St W, Toronto, M5V 2L7, Canada}

\date{\today}

\begin{abstract}
This work presents a quantum algorithm for the Monte Carlo pricing of financial derivatives. We show how the relevant probability distributions can be prepared in quantum superposition, the payoff functions can be implemented via quantum circuits, and the price of financial derivatives can be extracted via quantum measurements. We show how the amplitude estimation algorithm can be applied to achieve a quadratic quantum speedup in the number of steps required to obtain an estimate for the price with high confidence. This work provides a starting point for further research at the interface of quantum computing and finance. 
\end{abstract}
\maketitle 

\section{Introduction}

A great amount of computational resources are employed by participants in today's financial markets. 
Some of these resources are spent on the pricing and risk management of financial assets and their derivatives.
Financial assets include the usual stocks, bonds, and commodities, based upon which more complex contracts such as financial derivatives \cite{Hull2012} are constructed. 
Financial derivatives are contracts that have a future payoff dependent upon the future price or the price trajectory of one or more underlying benchmark assets. For these derivatives, due to the stochastic nature of underlying assets, an important issue is the assignment of a fair price based on available information from the markets, what in short can be called the \textit{pricing} problem \cite{Shreve2004,Follmer2004}. The famous Black-Scholes-Merton (BSM) model \cite{Black1973,Merton1973} can price a variety of financial derivatives via a simple and analytically solvable model that uses a small number of input parameters. A large amount of research has been devoted to extending the BSM model to include complicated payoff functions and complex models for the underlying stochastic asset  dynamics. 

Monte Carlo methods have a long history in the sciences. Some of the earliest known applications were by Ulam, von Neumann, Teller, Metropolis \textit{et al.}~\cite{eckhardt1987stan} in the context of the Los Alamos project, which used early computational devices such as the ENIAC.
For the pricing problem in finance, the main challenge is to compute an expectation value of a function of one or more underlying stochastic financial assets. For models beyond BSM, such pricing is often performed via Monte Carlo evaluation \cite{Glasserman2003}. 

Quantum computing promises algorithmic speedups for a variety of tasks, such as factoring or optimization. One of the earliest proposed algorithms, known as Grover's search \cite{Grover1996},  developed in the mid 1990s, in principle allows for a quadratic speed-up of searching an unstructured database. To find the solution in a size $N$ database with high probability, a classical computer takes $\Ord{N}$ computational steps, while  a quantum computer takes $\Ord{\sqrt{N}}$ steps. This algorithm has been extended and generalized to function optimization \cite{Durr1996}, amplitude amplification and estimation \cite{Brassard2002}, integration \cite{Heinrich2002}, quantum walk-based methods for element distinctness \cite{Ambainis2007}, and Markov chain algorithms \cite{Szegedy2004,Wocjan2009}, for example. In particular, the amplitude estimation algorithm can provide close to quadratic speedups for estimating expectation values \cite{Xu2018,Giovannetti2006,Magniez2007,Knill2007,Somma2008,Poulin2009,Chowdhury2017}, and thus provides a speedup to a problem for which Monte Carlo methods are used classically~\cite{Montanaro2015,Xu2018}.

Understanding the applications and enhancements of quantum mechanics to computational finance is still in its relative infancy.  The framework of quantum field theory can be harnessed to study the evolution of derivatives~\cite{Baaquie2004}. More recent works focus on the application of quantum machine learning~\cite{de2018advances, halperin2017qlbs} and quantum annealing~\cite{rosenberg2016solving} to areas such as portfolio optimization~\cite{elsokkary2017financial} and currency arbitrage~\cite{rosenberg2016finding}. This work investigates a new perspective of how to use quantum computing for the pricing problem.
We combine well-known quantum techniques, such as amplitude estimation~\cite{Brassard2002} and the quantum algorithm for  Monte Carlo~\cite{Montanaro2015,Xu2018}
with the pricing of financial derivatives. We first show how to obtain the expectation value of a financial derivative as the output of a quantum algorithm. To this end, we show the ingredients required to set up the financial problem on a quantum computer: the elementary arithmetic operations to compute payoff functions, the preparation of the model probability distributions used in finance, and the ingredients for estimating the expectation value through an imprinted phase on ancilla qubits.  It is shown how to obtain the quadratic speedup via the amplitude estimation algorithm. We discuss the quantum resources required to price European and Asian call options, representing fundamental types of derivatives. We provide evidence using classical numerical calculations that a quadratic speedup in pricing can be attained. 

This article begins with a brief summary of the basics of derivative pricing. The Black-Scholes-Merton framework is introduced in Section \ref{sectionBS} and classical Monte Carlo estimation is discussed in Section \ref{sectionMC} within the context of finance. 
In Section \ref{sectionQuantumMC}, the quantum algorithm for Monte Carlo is given.
Section \ref{sectionQuantumCall} specializes this quantum algorithm to the pricing of a European 
call option. 
Section \ref{sectionAsian} discusses the pricing of Asian options. 

\section{Black-Scholes-Merton option pricing}
\label{sectionBS}

\begin{figure}
\includegraphics[width=1.0\columnwidth]{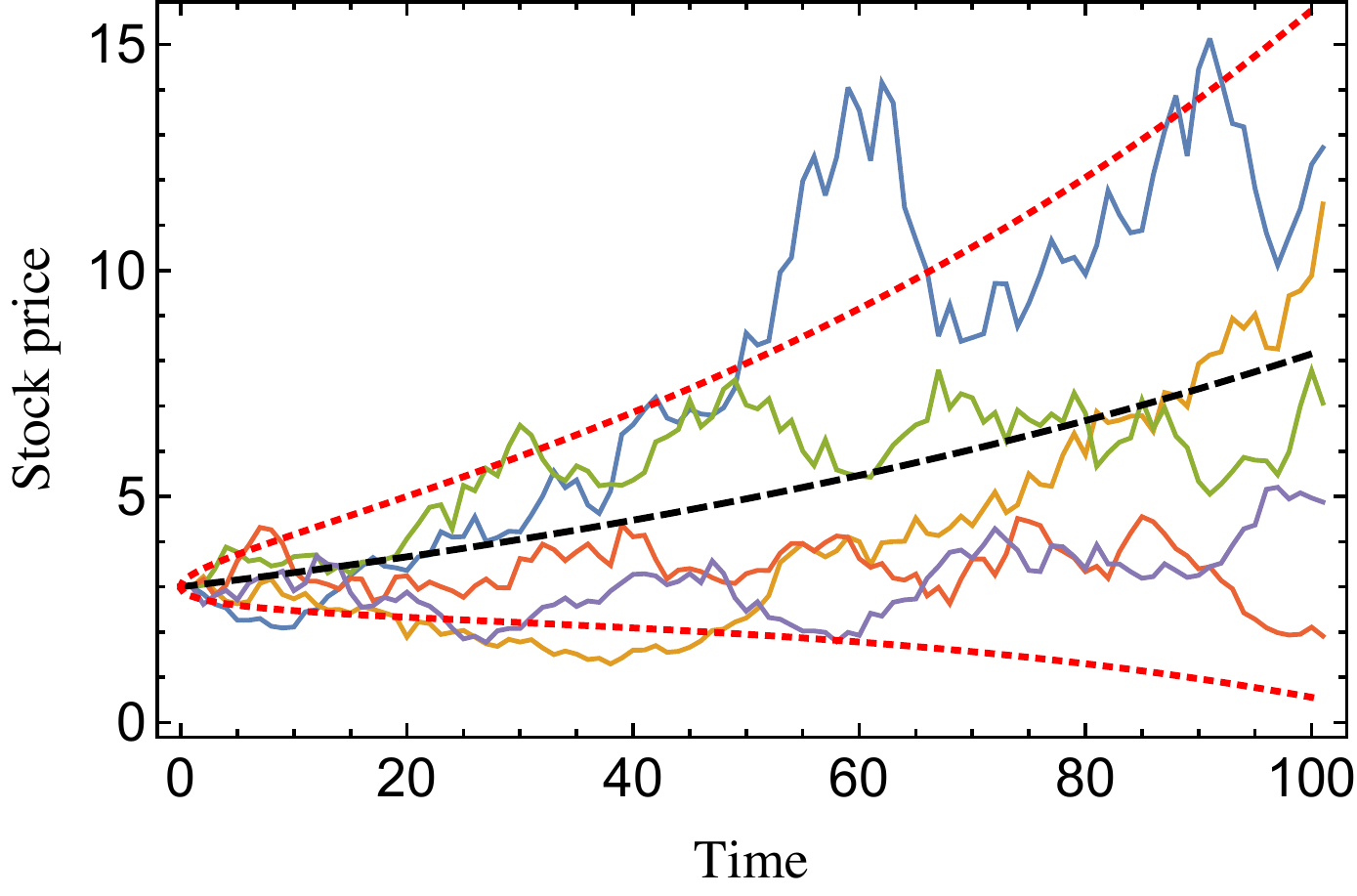}
	\caption{The price of a stock in the Black-Scholes-Merton model behaves stochastically as a geometric Brownian motion, changing each time step according to a log-normal distribution. Here, five sample price evolutions (in dollars) of a single stock are plotted as a function of time (in days). The resultant distribution is log-normally distributed, with mean (dashed black line) and one standard deviation (dotted red lines) illustrated. Pricing an option requires estimating the expected value of a payoff function on the stock at various times. The parameters are: initial price $S_0 =\$3$, drift $\alpha=0.1$, and volatility $\sigma=0.25$. }
	\label{figStock}
\end{figure}

The Black-Scholes-Merton (BSM) model \cite{Black1973,Merton1973} 
considers the pricing of financial derivatives (`options'). The original model assumes a single benchmark asset (`stock'), the price of which is stochastically driven by a Brownian motion, see Fig~\ref{figStock}. In addition, it assumes a risk-free investment into a bank account (`bond').
 We follow closely the discussion in the literature~\cite{Shreve2004} in the following.
\begin{defi} \label{defBM}
A Brownian motion $W_t$ is a stochastic process characterized by following attributes:
\begin{enumerate}
\item $W_0 = 0$; 
\item $W_{t}$ is continuous;
\item $W_{t}$ has independent increments;
\item $W_t - W_s \sim   N(0,t-s)\  {\rm for}\  t>s$.
\end{enumerate}
Here, $\mathcal N(\mu, \sigma^2)$ is the normal distribution with mean $\mu$ and
standard deviation $\sigma$. Point 3 means that the random variable $W_t - W_s$ for $t>s$
is independent of any previous time random variable $W_u$, $u<s$. The
probability measure under which $W_t$ is a Brownian motion shall be denoted by
$\mathbbm P$.
\end{defi}
The next step is to introduce a model for the market. 
\begin{defi}[Black-Scholes-Merton model] 
The Black-Scholes-Merton model consists of two assets, one risky (the stock), the other one risk-free (the bond). 
The risky asset is defined by the stochastic differential equation for the price dynamics given by
\begin{equation}
dS_t = S_t \alpha dt +  S_t \sigma dW_t,
\end{equation}
where $\alpha$ is the drift, $\sigma$ the volatility and $dW_t$ is a Brownian increment. The initial condition is $S_0$. In addition, the risk-free asset dynamics is given by
\begin{equation}
dB_t = B_t r dt,
\end{equation}
where $r$ is the risk-free rate (market rate). Set $B_0=1$. This model assumes that all parameters are constant, both assets can be bought or sold continuously and in unlimited and fractional quantities without transaction costs. Short selling is allowed, and the stock pays no dividends.
\end{defi}
Using Ito's lemma and the fact that $dW_t$ contributes an additional term in first order (due do its quadratic variation being proportional to $dt$), the risky asset stochastic differential equation can be solved as
\begin{equation}
S_t = S_0 e^{\sigma W_t +(\alpha -\sigma^2/2) t},
\end{equation}
see Appendix \ref{appendixBS}. Figure~\ref{figStock} shows sample evolutions of $S_t$. The risk-free asset is solved easily as
\begin{equation}
B_t = e^{r t}.
\end{equation}
This risk-free asset also is used for `discounting', i.e.~determining the present value of a future amount of money.
Let the task be to price an option. One of the simplest options is the 
European call option. The European call option gives the owner of the option the right to buy the stock at time $T\geq 0$ for  a pre-agreed price $K$.
\begin{defi} [European call option]
The European call option payoff is defined as
\begin{equation}\label{Eq:CallPay}
f(S_T) = \max \{ 0, S_T - K\},
\end{equation}
where $K$ is the strike price and $T$ the maturity date.
\end{defi}
The task of pricing is to evaluate at present time $t=0$ the expectation value of the option $f(S_T)$ on the stock on the maturity date.
The major tenet of risk-neutral derivative pricing is that the pricing is performed under a probability measure that shall not allow for arbitrage \cite{Follmer2004}. Simply put, arbitrage is a portfolio that has, at present, an expected future value that is greater than the current price of that portfolio. 
In the Black-Scholes-Merton framework, the stock price has a drift $\alpha$ under the $\mathbbm P$ measure. Any $\alpha \neq r$ allows for arbitrage under the measure $\mathbbm P$. When $\alpha > r$, one can make a profit above the market rate $r$ by investing in the stock, and when $\alpha < r$ one can make a profit by short selling the stock. Pricing of derivatives is performed under a probability measure where the drift of the stock price is exactly the market rate $r$. This pricing measure 
is denoted by $\mathbbm Q$ in contrast to the original measure $\mathbbm P$.

More formally, the probability measure $\mathbbm Q$ is defined such that the discounted asset price is a martingale, i.e.~the discounted expected value of the future stock price is the present day stock price itself. 
The martingale property is given in this context by 
\begin{equation}
  S_0 = e^{-r T} \mathbbm E_{\mathbbm Q} [S_T].
\end{equation}
Here, $e^{-r T}$ is the discount factor, which determines the present value of the payoff at a future time,
given the model assumption of a risk-free asset growing with $r$.
In addition, $\mathbbm E_{\mathbbm Q}[\cdot]$ denotes the $t=0$ expectation value under the measure $\mathbbm Q$.
Under this measure, investing in the stock does not, on average, return money above or below the market rate $r$, i.e.~does not allow for arbitrage. This feature is reflected in
the martingale price dynamics, which is given by 
\begin{equation}
dS_t = S_t r dt +  S_t \sigma d\tilde W_t,
\end{equation}
with the solution
\begin{equation}
S_t = S_0 e^{\sigma \tilde W_t +(r -\sigma^2/2) t}.
\end{equation}
Here, $\tilde W_t$ is a Brownian motion according to Definition \ref{defBM} under the martingale measure $\mathbbm Q$. The martingale property of $S_t $ is shown in Appendix \ref{appendixBS}, Lemma \ref{lemmaMartingale}.

The pricing problem is thus given by evaluating the risk-neutral price 
\begin{equation}\label{eqPrice}
\Pi = e^{-r T} \mathbbm E_{\mathbbm Q} [f(S_T)],
\end{equation}
which is the quantity of interest in this paper. For the simple European call option and several other options one can analytically solve the Black-Scholes-Merton model. A proof is sketched in Appendix \ref{appendixBS}.
\begin{result} [Black-Scholes-Merton price] \label{resultBS}
The risk-neutral price of the call option in Eq.~(\ref{Eq:CallPay}) is given by 
\begin{equation}
\Pi = \Phi(d_1) S_0 - \Phi(d_2) Ke^{-rT},
\end{equation}
with
\begin{eqnarray}
d_1&=& \frac{1}{\sigma \sqrt{T}} \left[\log \left( \frac{S_0}{K} \right) + \left(r+\frac{\sigma^2}{2} \right) T \right], \\
d_2&=& d_1 - \sigma \sqrt T,
\end{eqnarray}
and the cumulative distribution function of the normal distribution $p(x)$,
\begin{equation}
\Phi(x) =  \int_{-\infty}^x dy\ p(y) := \frac{1}{\sqrt {2\pi}} \int_{-\infty}^x dy\ e^{-\frac{y^2}{2}}.
\end{equation}
\end{result}

In the case of complex payoff functions  and/or complex asset price
dynamics, options prices cannot be solved analytically and one often resorts to Monte Carlo evaluation. Nevertheless, analytical solutions as above can be used for benchmarking Monte Carlo simulations.
Finally, note that there is a dynamical equation of motion for the options price $\Pi_t$ in the interval $0\leq t \leq T$, which is given by a partial differential equation. Such an equation is used in practice for `hedging', i.e.~safeguarding during the time $0\leq t \leq T$ while the option is active against eventual payouts. In this work we do not consider such a differential equation but focus on the present-day options price $\Pi \equiv \Pi_0$.

\section{Classical Monte Carlo pricing}
\label{sectionMC}

We first provide a brief overview of Monte Carlo derivative pricing. Options are usually nonlinear functions applied to the outcomes of one or multiple underlying assets. The option payoff depends on the asset prices at specific time instances or is based on the paths of the asset prices. As discussed in the previous section, 
European options depend on the asset price at a single future time. If the
nonlinear function is piecewise linear, the option can be priced analytically,
similar to Result \ref{resultBS} and Appendix \ref{appendixBS}.
If there are multiple
independent Brownian processes underlying the dynamics, the price can often also be determined analytically. The need for Monte Carlo arises if the payoff function is nonlinear beyond piecewise linear or, for example, in cases when different asset prices are assumed to be correlated. 
Another class of options, called American options, allow the buyer to exercise the option at any point in time between the option start and the option maturity. Such options are related to the optimal stopping problem \cite{Follmer2004} and are also priced using Monte Carlo methods \cite{Longstaff2001}. 
Asian options depend on the average asset price during a time intervals and, if the averaging is arithmetic, may also require Monte Carlo \cite{Kemna1990}. 

The underlying asset prices are modeled via stochastic differential equations.
Often stock prices are taken to be log-normal stochastic processes, i.e.,~driven
by an exponentiated Brownian motion. In this case, when the parameters are constant, the stochastic differential
equation is exactly solvable. In other cases, such as when the parameters of the
model such as the volatility itself follow a
stochastic differential equation, the asset price dynamics is usually not
analytically solvable. The price is then determined by sampling paths of the asset dynamics. Moreover, Brownian motions are fundamentally continuous and Gaussian with exponentially suppressed tails, features which are rarely observed in real markets. Further research has considered `fat-tailed' stochastic processes and Levy jump processes \cite{Schoutens2003}, which often also require Monte Carlo sampling. 

Monte Carlo pricing of financial derivatives proceeds in the following way. Assume that the risk-neutral probability distribution is known, or can be obtained from calibrating to market variables.
Sample from this risk-neutral probability distribution
a market outcome, compute the asset prices given that market outcome, then compute the option payoff given the asset prices. Averaging the payoff over multiple samples obtains an approximation of the derivative price. 
Assume a European option on a single benchmark asset and let the true option price be $\Pi$ and $\hat \Pi$ be the approximation obtained from $k$ samples. 
Assume that the random variable of the payoff $f(S_T)$ is bounded in variance, i.e.~$\mathbbm V[f(S_T)] \leq \lambda^2$. Then the probability 
that the price estimation $\hat \Pi$ is $\epsilon$ away from the true price
is determined  by Chebyshev's inequality \cite{Montanaro2015}
\begin{equation} 
\mathbbm P[\vert \hat \Pi - \Pi \vert \geq \epsilon] \leq \frac{\lambda^2}{k \epsilon^2}.
\end{equation}
For a constant success probability, we thus require 
\begin{equation}\label{eqStepsClassical}
k= \Ord{ \frac{\lambda^2}{\epsilon^2}}
\end{equation} 
samples to estimate to additive error $\epsilon$. The task of the quantum algorithm will be to improve the $\epsilon$ dependence from 
$\epsilon^2$ to $\epsilon$, hence providing a quadratic speedup for a given error.

Before we discuss the quantum algorithm for derivative pricing, we show how to encode expectation values into a quantum algorithm and how to obtain the same $\epsilon$ dependency as the classical algorithm. 
We then discuss the quadratic speedup by using the fundamental quantum algorithm of amplitude estimation. 

\section{Quantum algorithm for Monte Carlo}
\label{sectionQuantumMC}

We first discuss generically the quantum algorithms to measure an expectation value and to obtain a quadratic improvement in the number of measurements \cite{Brassard2002,Montanaro2015,Xu2018}. See Appendix~\ref{appendixQIntro} for a brief introduction to the neccessary elements of quantum mechanics. In the following sections, we then specialize to European and Asian options. 
Assume we are given an algorithm $\mathcal A$ on $n$ qubits (the subsequent discussion can also be generalized to measuring only a
subset of qubits \cite{Montanaro2015}). 
When measuring the $n$
qubits, the algorithm produces the $n$-bit string result $x$ with probability $\vert
a_x\vert^2$. In addition let $v(x)$ be a function 
$ v(x): \{0,1\}^n \to \mathbbm R$ mapping from $n$-bit strings to reals. Here, $v(\mathcal A)$ denotes the random variable specified by the algorithm $\mathcal A$ and the function $v(x)$.
The task is to obtain the expectation value 
\begin{equation}
\mathbbm E[v(\mathcal A)] := \sum_{x=0}^{2^n-1} \vert a_x\vert^2 v(x) .
\end{equation}
In addition, assume we can implement a rotation onto an ancilla qubit,
\begin{equation} \label{eqOptionRotation}
\mathcal R \ket x \ket 0 = \ket x (\sqrt{1-v(x)} \ket 0 +\sqrt{v(x)} \ket 1).
\end{equation} 

These elements are now combined into a simple quantum algorithm to obtain the expectation value.
First apply the algorithm $\mathcal A$:
\begin{equation}
\mathcal A \ket {0^n} = \sum_{x=0}^{2^n-1} a_x \ket x ,
\end{equation}
where $\ket {0^n}$ denotes the $n$ qubit register with all qubits in the state $\ket 0$. Then perform the rotation of an ancilla via $\mathcal R$
\begin{eqnarray}
&\sum_{x=0}^{2^n-1}  a_x \ket x \ket{0 } \\ &\to \sum_{x=0}^{2^n-1} a_x \ket x (\sqrt{1-v(x)} \ket 0 +\sqrt{v(x)} \ket 1) =: \ket \chi. \nonumber
\end{eqnarray}
Combining the two operations defines a unitary $\mathcal F$ and the resulting state $\ket \chi$
\begin{equation}
\mathcal F \ket{0^{n+1}} := \mathcal R \mathcal (\mathcal A\otimes \mathcal I_2) \ket{0^{n+1}} \equiv \ket \chi.
\end{equation}
Here, $\mathcal I_d$ is the $d$-dimensional identity operator.
Measuring the ancilla in the state $\ket 1$ obtains as the success probability the expectation value
\begin{equation}
\mu:=\bra \chi \left (\mathcal I_{2^n} \otimes \ket 1 \bra 1 \right ) \ket \chi = \sum_{x=0}^{2^n-1} \vert a_x \vert^2 v(x) \equiv \mathbbm E[v(\mathcal A)]. 
\end{equation}
This success probability can be obtained by repeating the procedure $t$ times 
and collecting the clicks for the $\ket 1$ state as a fraction of the total measurements. 
The variance is $\epsilon^2= \frac{\mu(1-\mu)}{t}$ from the Bernoulli distribution, i.e.~the standard deviation
is $\epsilon= \sqrt{\frac{\mu(1-\mu)}{t}}$. Hence, the experiment has to 
be repeated 
\begin{equation} \label{eqStepsQuantumNaive}
t = \Ord{ \frac{\mu(1-\mu)}{\epsilon^2} }
\end{equation}
times for a given accuracy $\epsilon$. This quadratic dependency in $\epsilon$ is analogous to the classical Monte Carlo dependency Eq.~(\ref{eqStepsClassical}). Obtaining a quadratic speedup for the number of repetitions is the core task of amplitude estimation. 

\begin{figure*}
\begin{center}
\begin{minipage}{0.09\textwidth}
\flushleft \vspace{-0.0cm}\hspace{-1.0cm} \bf{$\qquad$(a)} \vspace{-0.35cm}\\
$$
\Qcircuit @C=1em @R=.0001em {
& \multigate{5}{Q} & \qw \\
& \ghost{Q} & \qw \\
& \ghost{Q} & \qw \\
& \ghost{Q} & \qw \\
& \ghost{Q} & \qw \\
& \ghost{Q} & \qw \\
}
$$
\end{minipage}
\begin{minipage}{0.04\textwidth}
\vspace{0.65cm}
$\Leftrightarrow$
\end{minipage}
\begin{minipage}{0.46\textwidth}
\vspace{0.3cm}
$$
\Qcircuit @C=1em @R=0.18em {
& \multigate{5}{\mathcal{F}} & \multigate{5}{\mathcal Z} & \multigate{5}{\mathcal{F}^{\dagger}} & \multigate{5}{\mathcal V} & \multigate{5}{\mathcal{F}} & \multigate{5}{\mathcal Z} & \multigate{5}{\mathcal{F}^{\dagger}} & \multigate{5}{\mathcal V} & \qw \\
& \ghost{\mathcal{F}} & \ghost{\mathcal Z} & \ghost{\mathcal{F}^{\dagger}} & \ghost{V} & \ghost{\mathcal{F}} & \ghost{\mathcal Z} & \ghost{\mathcal{F}^{\dagger}} & \ghost{V} & \qw \\
& \ghost{\mathcal{F}} & \ghost{\mathcal Z} & \ghost{\mathcal{F}^{\dagger}} & \ghost{V} & \ghost{\mathcal{F}} & \ghost{R} & \ghost{\mathcal{A}^{\dagger}} & \ghost{V} & \qw \\
& \ghost{\mathcal{F}} & \ghost{\mathcal Z} & \ghost{\mathcal{F}^{\dagger}} & \ghost{V} & \ghost{\mathcal{F}} & \ghost{R} & \ghost{\mathcal{A}^{\dagger}} & \ghost{V} & \qw \\
& \ghost{\mathcal{F}} & \ghost{\mathcal Z} & \ghost{\mathcal{F}^{\dagger}} & \ghost{V} & \ghost{\mathcal{F}} & \ghost{R} & \ghost{\mathcal{A}^{\dagger}} & \ghost{V} & \qw \\
& \ghost{\mathcal{F}} & \ghost{\mathcal Z} & \ghost{\mathcal{F}^{\dagger}} & \ghost{V} & \ghost{\mathcal{F}} & \ghost{R} & \ghost{\mathcal{A}^{\dagger}} & \ghost{V} & \qw \\
}
$$
\end{minipage}
\begin{minipage}{0.39\textwidth}
\flushleft \vspace{0.3cm}\hspace{-0.6cm} \bf{$\qquad$(b)} \vspace{0.0cm}\\
\includegraphics[width=\columnwidth]{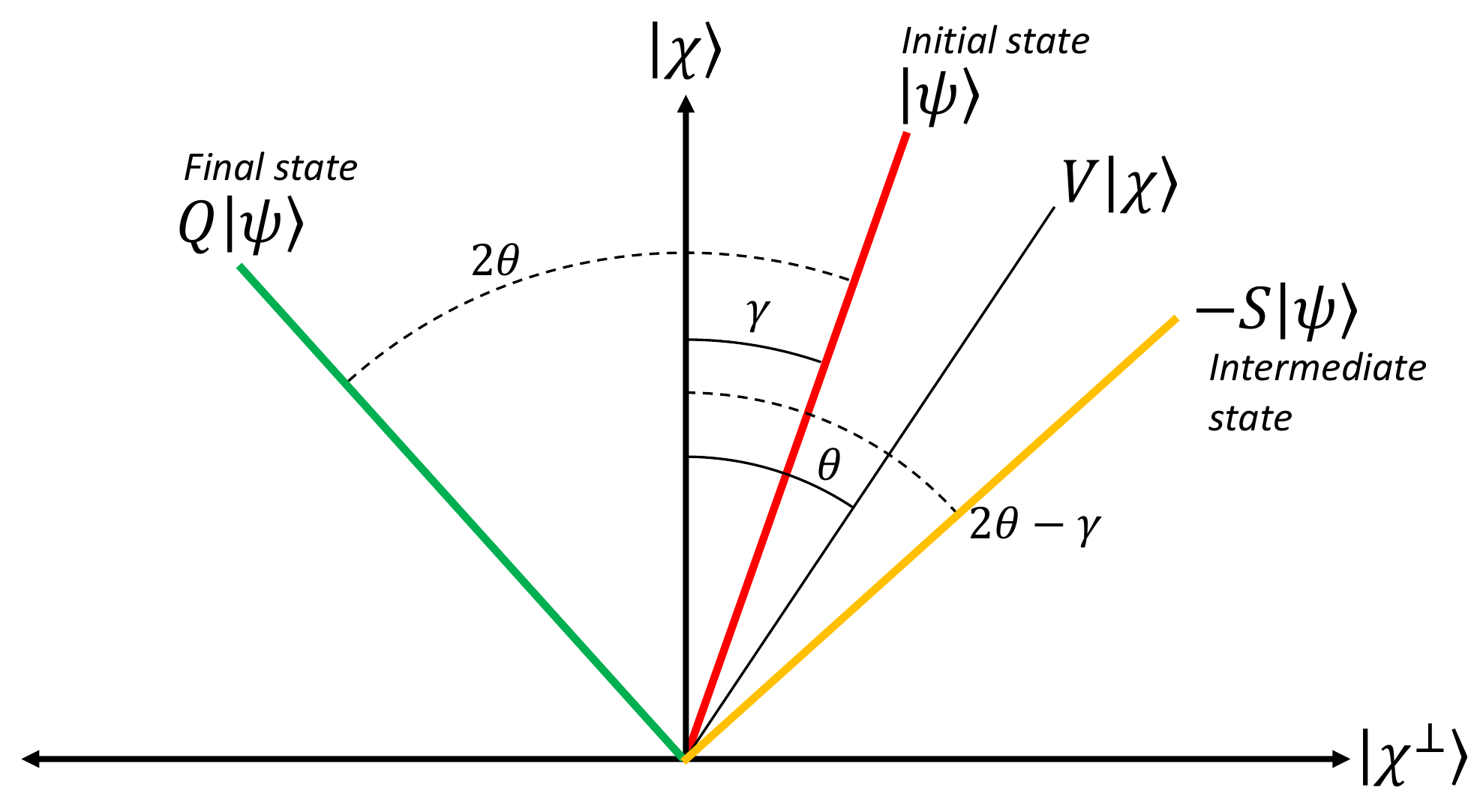}
\end{minipage}
\begin{minipage}{0.52\textwidth}
\flushright {\bf(c)$\qquad$}
\vspace{-0.5cm}
$$
\Qcircuit @C=1em @R=.3em {
& & \qquad \mbox{Preparing $\ket{\chi}$} & & & & \quad\quad \mbox{controlled amplitude amplification} \\
\\
\\
& \lstick{\ket{0}} & \multigate{5}{\mathcal A} & \multigate{6}{\mathcal R} & \multigate{6}{\mathcal Q} & \multigate{6}{\mathcal Q^{2}} & \qw & \ldots & & \multigate{6}{\mathcal Q^{2^{n-1}}} & \qw & \qw \\
& \lstick{\ket{0}} & \ghost{\mathcal A} & \ghost{\mathcal R} & \ghost{\mathcal Q} & \ghost{\mathcal Q^{2}} & \qw & \ldots & & \ghost{\mathcal Q^{2^{n-1}}} & \qw & \qw \\
& \lstick{\ket{0}} & \ghost{\mathcal A} & \ghost{\mathcal R} & \ghost{\mathcal Q} & \ghost{\mathcal Q^{2}} & \qw & \ldots & & \ghost{\mathcal Q^{2^{n-1}}} & \qw & \qw \\
& \lstick{\ket{0}} & \ghost{\mathcal A} & \ghost{\mathcal R} & \ghost{\mathcal Q} & \ghost{\mathcal Q^{2}} & \qw & \ldots & & \ghost{\mathcal Q^{2^{n-1}}} & \qw & \qw \\
& \lstick{\ket{0}} & \ghost{\mathcal A} & \ghost{\mathcal R} & \ghost{\mathcal Q} & \ghost{\mathcal Q^{2}} & \qw & \ldots & & \ghost{\mathcal Q^{2^{n-1}}} & \qw & \qw \\
& \lstick{\ket{0}} & \ghost{\mathcal A} & \ghost{\mathcal R} & \ghost{\mathcal Q} & \ghost{\mathcal Q^{2}} & \qw & \ldots & & \ghost{\mathcal Q^{2^{n-1}}} & \qw & \qw \\
& \lstick{\ket{0}} & \qw & \ghost{W} & \ghost{\mathcal Q} & \ghost{\mathcal Q^{2}} & \qw & \ldots & & \ghost{\mathcal Q^{2^{n-1}}} & \qw & \qw & \\
\\
\\
\\
\\
\\
\\
& \lstick{\ket{0}} & \gate{\mathcal H} & \qw & \ctrl{-7} & \qw & \qw & \qw & \qw & \qw & \multigate{9}{QFT^{-1}} & \meter \\
& \lstick{\ket{0}} & \gate{\mathcal H} & \qw & \qw & \ctrl{-8} & \qw & \qw & \qw & \qw & \ghost{QFT^{-1}} &  \meter \\
\\
\\
& \lstick{\vdots} & & & & & & \ddots & & & & \vdots \\
\\
\\
\\
\\
& \lstick{\ket{0}} & \gate{\mathcal H} & \qw & \qw & \qw & \qw & \qw & \qw & \ctrl{-16} & \ghost{QFT^{-1}} &  \meter \gategroup{4}{2}{10}{4}{0.7em}{.} \gategroup{4}{5}{26}{10}{.7em}{--} \\
}
$$
\end{minipage}
\begin{minipage}{0.47\textwidth}
\flushleft \vspace{0.0cm} \bf{$\qquad$(d)} \vspace{0.3cm}\\
\includegraphics[width=\columnwidth]{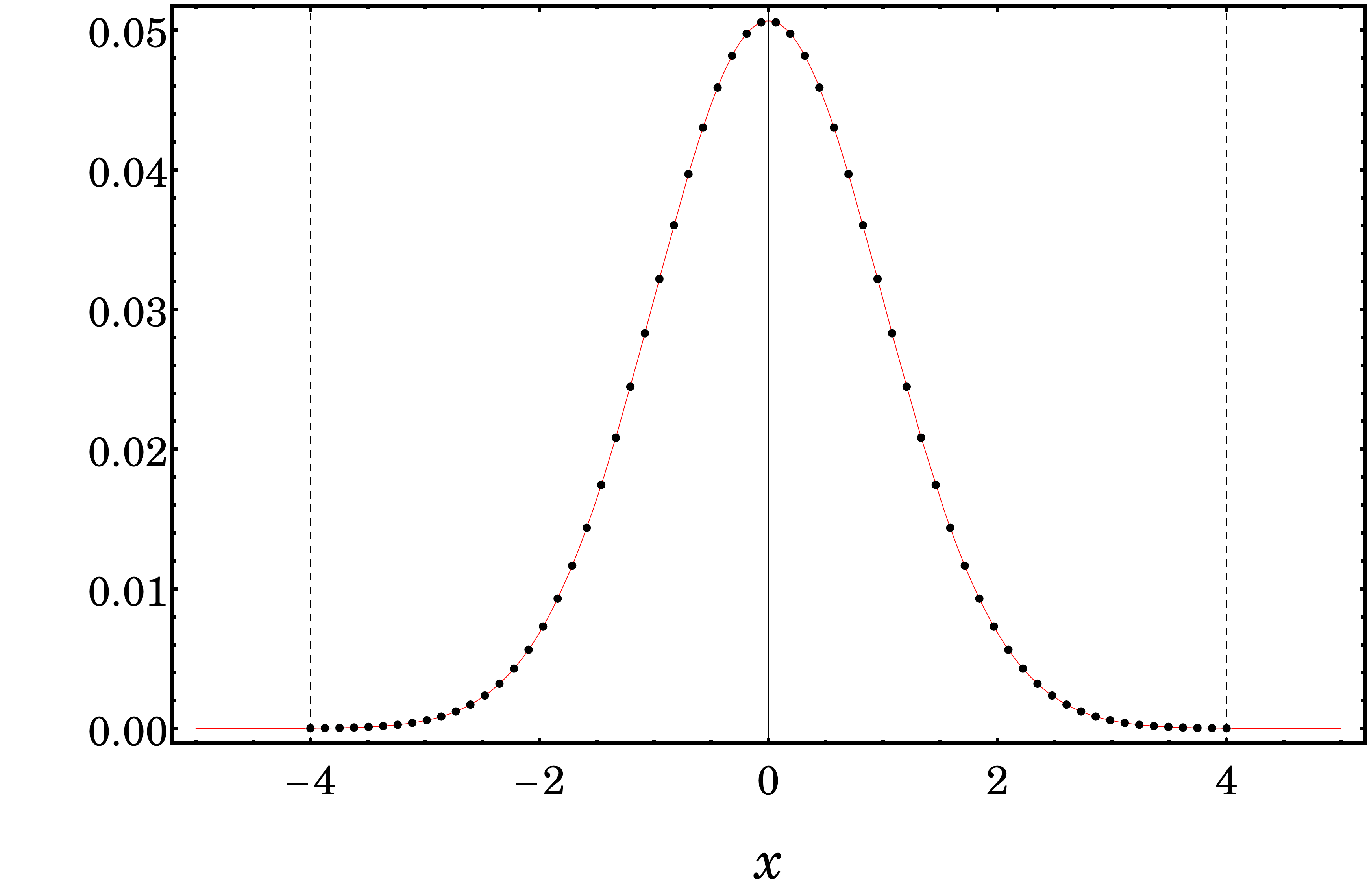}
\end{minipage}
\end{center}
\caption{Using amplitude estimation for the quantum Monte Carlo pricing of financial derivatives. {\bf(a)} The $n+1$ qubit phase estimation unitary is written in terms of $\mathcal{F}:= \mathcal R (\mathcal A \otimes \mathcal I_{2})$, and the simple rotation unitaries $\mathcal Z:=\mathcal I_{2^{n+1}}-2\ket{0^{n+1}}\bra{0^{n+1}}$ and $\mathcal V:=\mathcal I_{2^{n+1}}-2 \mathcal I_{2^{n}}\otimes \ket{1}\bra{1}$. {\bf(b)} A visualization of the action of $\mathcal{Q} := \mathcal{U}\mathcal{S}$, with $\mathcal{S} = \mathcal{V}\mathcal{U}\mathcal{V}$ and $\mathcal U = \mathcal{F}\mathcal{Z}\mathcal{F}^{\dagger}$, on an arbitrary state $\ket{\psi}$ (red) in the span of $\ket{\chi}$ and $\mathcal{V}\ket{\chi}$. First, the action of $-\mathcal S$ on $\ket{\psi}$ is to reflect along $\mathcal V\ket{\chi}$, resulting in the intermediate $-\mathcal S\ket{\psi}$ (amber). Then, $-\mathcal U$ acts on $-\mathcal S\ket{\psi}$ by reflecting along $\ket{\chi}$. The resultant state $\mathcal Q\ket{\psi}$ (green) has been rotated anticlockwise by an angle $2\theta$ in the hyperplane of $\ket{\chi}$ and $\ket{\chi^{\perp}}$. {\bf(c)}~The phase estimation circuit. Here, $\mathcal A$ encodes the randomness by preparing a superposition in $\ket{x}$, while $\mathcal R$ encodes the random variable into the $\ket{1}$ state of an ancilla qubit according to Eq.~(\ref{eqOptionRotation}). The output after both steps is the multiqubit state $\ket{\chi}$. Amplitude estimation then proceeds by invoking phase estimation to encode the rotation angle $\theta$ in a register of quantum bits that are measured to obtain the estimate $\hat \theta$. {\bf(d)} For pricing a European call option, the superposition prepared by $\mathcal A$ (or equivalently $\mathcal G$ in Eq.~\eqref{eqGroverState}) is a discretization of the normal distribution in $x$ with a fixed cutoff (e.g. $c=4$), approximating the Brownian motion of the underlying asset.
In this case, $\mathcal R$ encodes the call option payoff in Eq.~(\ref{Eq:CallPay}). }
\label{figure1}
\end{figure*}

The main tool to obtain a quantum speedup is to connect the desired expectation value to an eigenfrequency of an oscillating quantum system and then use another quantum degree of freedom (such as another register of qubits) as a probe to extract the eigenfrequency.
Note that we can slightly redefine the quantity being measured. Define the unitary 
\begin{equation}
\mathcal V:= \mathcal I_{2^{n+1}} -2\mathcal I_{2^n}\otimes \ket 1 \bra 1,
\end{equation}
for which $\mathcal V= \mathcal V^\dagger$ and $\mathcal V^2 = \mathcal I_{2^{n+1}}$. A measurement of $\mathcal V$ on $\ket \chi$ obtains $\bra \chi \mathcal V \ket \chi = 1-2 \mathbbm \mu$. From this measurement we can extract the desired expectation value. 

Any quantum state in the $(n+1)$-qubit Hilbert space can be expressed as a linear combination of $\ket \chi$ and a specific orthogonal complement $ \ket{\chi^\perp}$. Thus, we can express 
$\mathcal V \ket \chi = \cos(\theta/2) \ket \chi + e^{i\phi} \sin(\theta/2) \ket{\chi^\perp}$, with the angles $\phi$ and $\theta$. Note that our expectation value can be retrieved via 
\begin{equation}\label{eqMeanToAngle}
1-2 \mu = \cos(\theta/2).
\end{equation}
The task becomes to measure $\theta$. We now define a transformation $\mathcal Q$ that encodes $\theta$ in its eigenvalues.
First, define the unitary reflection
\begin{equation}
\mathcal U:= \mathcal I_{2^{n+1}} -2 \ket \chi \bra \chi,
\end{equation}
which acts as $\mathcal U \ket \chi = - \ket \chi$ and  $\mathcal U \ket {\chi^\perp} = \ket {\chi^\perp}$ for any orthogonal state. 
Note that $-\mathcal U$ reflects across $\ket \chi$ and leaves $\ket {\chi}$ itself unchanged.
This unitary can be implemented as $\mathcal U=\mathcal F \mathcal Z \mathcal F^\dagger$, where $\mathcal F^\dagger $ is the inverse of $\mathcal F$ and $\mathcal Z :=\mathcal I_{2^{n+1}}-2 \ket{0^{n+1}} \bra{0^{n+1}}$ is the reflection of the computational zero state. 
Similarly, define the unitary 
\begin{equation}
\mathcal S:=\mathcal I_{2^{n+1}}-2 \mathcal V \ket \chi \bra \chi \mathcal V \equiv \mathcal  V \mathcal U \mathcal V.
\end{equation}
Note that $-\mathcal S$ reflects across $\mathcal V\ket \chi$ and leaves $\mathcal V\ket \chi$ itself unchanged. 
The transformation 
\begin{equation}
\mathcal Q := \mathcal U \mathcal S = \mathcal U \mathcal V \mathcal U \mathcal V
\end{equation}
performs a rotation by an angle $2\theta$ in the two-dimensional Hilbert space spanned by $\ket \chi$ and $ \mathcal V\ket{\chi }$. Figure~\ref{figure1} (a) shows the breakdown of $\mathcal{Q}$ into its constituent unitaries and Fig.~\ref{figure1}~(b) illustrates how $\mathcal{Q}$ imprints a phase of $2\theta$ via the reflections just discussed. The eigenvalues of $\mathcal Q$ are $e^{\pm i \theta}$ with corresponding eigenstates $\ket{\psi_\pm}$ \cite{Xu2018}. The task is to resolve these eigenvalues via phase estimation, as shown in Fig.~\ref{figure1} (c). 

For phase estimation of $\theta$~\cite{nielsen2002quantum}, we require the conditional application of the operation $\mathcal Q$. 
Concretely, we require
\begin{equation} \label{eqQc}
\mathcal Q^c: \ket j \ket \psi \to \ket j \mathcal Q^j \ket \psi, 
\end{equation}
for an arbitrary $n$ qubit state $\ket \psi$. Phase estimation then proceeds in the following way, see Fig.~\ref{figure1} (c). Take a copy of $\ket \chi$ by applying $\mathcal{F}$ to a register of qubits in $\ket{0^{n+1}}$. Then prepare an additional $m$-qubit register in the uniform superposition via the Hadamard operation $\mathcal H$
\begin{equation}
\mathcal H^{\otimes m} \ket{0^m} \ket \chi = \frac{1}{\sqrt{2^m}} \sum_{j=0}^{2^m-1} \ket j \ket \chi.
\end{equation}
Then perform the controlled operation $\mathcal Q^c$ to obtain
\begin{equation} \label{eqControlledQ}
\frac{1}{\sqrt{2^m}}  \sum_{j=0}^{2^m-1} \ket j \mathcal Q^j \ket \chi.
\end{equation}
One can show that  
$\ket \chi = \frac{1}{\sqrt 2} \left( \ket {\psi_+}  + \ket {\psi_-} \right)$ is the expansion of $\ket \chi$ into the two eigenvectors of $\mathcal Q$ corresponding to the eigenvalues $e^{\pm i \theta}$ \cite{Xu2018}.  An inverse quantum Fourier transformation applied to Eq.~(\ref{eqControlledQ})
prepares the state
\begin{equation}
\sum_{x=0}^{2^m-1} \alpha_+(x) \ket x \ket {\psi_+} +  \alpha_-(x) \ket x  \ket{\psi_-}.
\end{equation}
The $\vert \alpha_\pm(x)\vert^2$ are peaked where $x/2^m= \pm \hat \theta$ is an $m$-bit approximation to $\pm \theta$. 
Hence, measurement of the $\ket x $ register will retrieve the approximations $ \pm \hat \theta$.
The detailed steps are shown in Appendix \ref{appendixPhaseEstimation}.

These results can be formalized with the following theorems.
\begin{theorem} [Amplitude estimation \cite{Brassard2002}] \label{theoremAmpEst}
There is a quantum algorithm called amplitude
estimation which takes as input: one copy of a quantum state $\ket \chi$, a unitary transformation $\mathcal U = \mathcal I -2 \ket \chi \bra \chi$, a unitary transformation $\mathcal V =\mathcal I -2P$ for some projector $P$, and an integer $t$. The algorithm
outputs $\hat a$, an estimate of $a =  \bra \chi P \ket \chi$, such that
$$\vert \hat a - a \vert \leq 2 \pi \frac{\sqrt{a(1-a)} }{t} + \frac{\pi^2}{t^2}$$
with probability at least $8/\pi^2$, using $\mathcal U$ and $\mathcal V$ $t$ times each.
\end{theorem}
This theorem can be used to estimate expectation values. Given the cosine relationship in Eq.~(\ref{eqMeanToAngle}), it is natural to start with $[0,1]$ bounded expectation values. 
\begin{theorem} [Mean estimation for {$[0,1]$} bounded functions \cite{Montanaro2015}] \label{theoremMean01}
Let there be given a quantum circuit $\mathcal A$ on $n$ qubits. 
Let $v(\mathcal A)$ be the random variable that maps to $v(x) \in [0,1]$ when the bit string $x$ is measured as the output of $\mathcal A$.
Let $\mathcal R$ be defined as 
$$
\mathcal R \ket x \ket 0 = \ket x (\sqrt{1-v(x)} \ket 0 - \sqrt{v(x)} \ket 1).
$$
Let $\ket \chi$ be defined as $\ket \chi = \mathcal R (\mathcal A \otimes \mathcal I_2) \ket{0^{n+1}}$.
Set $\mathcal U = \mathcal I_{2^{n+1}} - 2 \ket \chi \bra \chi $. 
There exists a quantum algorithm that uses
$\Ord{\log 1/\delta}$ copies of the state $\ket \chi$, uses $\mathcal U$ for a number of times proportional to $\Ord{t \log 1/\delta}$ and outputs an estimate $\hat \mu$ such that
$$\vert \hat \mu- E[v(\mathcal A)] \vert \leq C \left( \frac{\sqrt{\mathbbm E [ v(\mathcal A)]}}{t} + \frac{1}{t^2}\right) $$
with probability at least $1-\delta$, where $C$ is a universal constant. In particular, for any fixed $\delta>0$ and any $\epsilon$
such that $0< \epsilon\leq 1$, to produce an estimate $\hat \mu$ such that with probability at least $1-\delta$, $\vert \hat \mu - E [ v(\mathcal A)] \vert \leq \epsilon  E [ v(\mathcal A)]$,
it suffices to take $t=\Ord{(1/(\epsilon \sqrt{E [ v(\mathcal A)]} )}$. To achieve 
$\vert \hat \mu - E [ v(\mathcal A)] \vert \leq \epsilon$ with probability at least 
$1-\delta$, it suffices to take $t=\Ord{1/\epsilon}$.
\end{theorem}
This theorem is a direct application of amplitude estimation via Theorem \ref{theoremAmpEst}. The success probability of $8/\pi^2$ of amplitude estimation can be improved to $1-\delta$ by taking the median of multiple runs of Theorem \ref{theoremAmpEst}, see Appendix \ref{appendixPhaseEstimation}, Lemma \ref{lemmaMedian}.
Theorem \ref{theoremMean01} can be generalized to random variables with bounded variance as follows.
\begin{theorem}[Mean estimation with bounded variance \cite{Montanaro2015}] \label{theoremMeanL2}
Let there be given a quantum circuit $\mathcal A$. 
Let $v(\mathcal A)$ be the random variable corresponding to $v(x)$ when the outcome $x$ of $\mathcal A$ is measured, 
such that $\mathbbm V [v(\mathcal A)] \leq \lambda^2$. Let the accuracy be $\epsilon < 4 \lambda$.
Take $\mathcal U$ and $\ket \chi$ as in Theorem \ref{theoremMean01}.
There exists
a quantum algorithm that uses $\Ord{\log(\lambda/\epsilon) \log \log (\lambda/\epsilon)}$ copies
of $\ket \chi$ and uses $\mathcal U$ for a number of times $\Ord{(\lambda/\epsilon) \log^{3/2}(\lambda/\epsilon) \log \log (\lambda/\epsilon)}$ and estimates $\mathbbm E[v(\mathcal A) ]$ up to additive error $\epsilon$ with
success probability at least $2/3$.
\end{theorem}
Theorem \ref{theoremMeanL2} can be proved by employing Theorem \ref{theoremMean01}. We proceed by sketching the proof, and refer the interested reader to the detailed treatment in \cite{Montanaro2015}.
To show Theorem \ref{theoremMeanL2}, 
the bounded-variance random variable $v(\mathcal A)$ is related to a set of 
random variables with outputs between $[0,1]$ and then the estimates of the mean of each of these random variables are combined to give the final estimate. This can be done in three steps. First, the random variable $v(\mathcal A)$ can be approximately standardized by subtracting an approximation to the mean and dividing by the known variance bound $\lambda^2$, to obtain a random variable $v'(\mathcal A)$. Second,  $v'(\mathcal A)$ can be split into positive and negative parts by using the functions $f_{\min}(x) = \min \{x,0\}$ and $f_{\max}(x) = \max \{x,0\}$. (Coincidentally, these function are similar to the call and put option payoff functions.) This defines new random variables 
$B_{<}=-f_{\min}(v'(\mathcal A))$ and  $B_{>}=f_{\max}(v'(\mathcal A))$, both taking on only values $\geq 0$. These random variables can be rescaled and combined to give the desired random variable. 

As the third step, both positive random variables $B_{<,>}=:B$ can be split into multiple auxiliary random variables with outputs between $[0,1]$ and each of these random variables is estimated by Theorem \ref{theoremMean01}. This is done by defining the functions for $0\leq a < b$
\begin{equation} 
f_{a,b}(x) = \frac{1}{b}
\begin{cases}
x & \text{if } a\leq x < b, \\
0              & \text{otherwise.}
\end{cases}
\end{equation} 
which take on values in $[0,1]$. Now the auxiliary random variables $f_{0,1}(B)$, $f_{1,2}(B)$, $f_{2,4}(B)$, $f_{4,8}(B)$ and so forth can be defined which are all taking values in $[0,1]$.  Theorem \ref{theoremMean01} can be used to estimate the mean of each of these random variables. It can be shown that only a small number $\lceil \log(\lambda/\epsilon) \rceil$ of these auxiliary random variables are needed 
to estimate the mean of random variable $B$ with final error $\epsilon$. This estimation makes use of the bounded-variance property which leads to the distribution tails contributing only a small error. 

The resource count of Theorem \ref{theoremMeanL2} is justified as follows. 
For an estimation with accuracy $\epsilon$ it can be shown that 
the number of random variables $f_{a,b}(B)$ and therefore the number of applications of Theorem \ref{theoremMean01} needed is $\lceil \log(\lambda/\epsilon) \rceil$. In addition, the number of steps $t$ in Theorem \ref{theoremMean01} can be taken to be $t =\Ord{(\lambda/\epsilon) \sqrt{\log \lambda/\epsilon}}$ and also $\delta = \Ord{1/\log (\lambda/\epsilon)}$ to achieve the final accuracy. Thus, from each application of the Theorem \ref{theoremMean01}, we need $\Ord{\log \log (\lambda/\epsilon)}$ copies of $\ket \chi$ and $\Ord{\frac{\lambda}{\epsilon} \sqrt{\log (\lambda/\epsilon)} \log \log (\lambda/\epsilon)}$ applications of $\mathcal U$. As we apply Theorem \ref{theoremMean01} for $\lceil \log(\lambda/\epsilon) \rceil$ times we require $\Ord{\log(\lambda/\epsilon) \log \log (\lambda/\epsilon)}$ copies of $\ket \chi$ and $\Ord{\frac{\lambda}{\epsilon} (\log (\lambda/\epsilon))^{3/2} \log \log (\lambda/\epsilon)}$ applications of $\mathcal U
$. 
This is an almost quadratic speedup in the number of applications of $\mathcal{U}$ when compared to the direct approach outlined in Eq.~\eqref{eqStepsQuantumNaive}.

\section{Quantum algorithm for European option pricing}
\label{sectionQuantumCall}

We specialize the above discussion to the Monte Carlo pricing of a European call option. We show how to prepare the Brownian motion distribution and the quantum circuit for the option payoff.
For any European option, i.e.~an option that depends on the asset prices only at a single maturity date $T$, we can write the price as an expectation value of a function of the underlying stochastic processes evaluated at the maturity date. For the BSM model with a single Brownian motion we have 
\begin{equation}\label{eqBSPriceEuro}
\Pi = e^{-rT} \mathbbm E_{\mathbbm Q}[v(W_T)],
\end{equation}
where $W_T$ is the Brownian motion at time $T$ and $v(x)$ is, for example, defined from 
the function $f(x)$ in Eq.~(\ref{Eq:CallPay}). 
From Definition \ref{defBM}, $W_T$ is a Gaussian random variable $\sim \mathcal N(0,T)$. The probability density for this random variable
is given by
\begin{equation}
p_T(x) = \frac{1}{\sqrt{2\pi T}} e^{-\frac{x^2}{2T} }.
\end{equation}
To prepare an approximate superposition of these probabilities, we take the support of this density from $[-\infty,\infty] \to [-x_{\max}, x_{\max}]$ and discretize this interval with $2^n$ points, where $n$ is an integer. 
Here, $x_{\max}=\Ord{\sqrt{T}}$, as a few standard deviations are usually enough to capture the normal distribution and reliably estimate the options price.
The discretization points may be defined as $x_j := -x_{\max} + j \Delta x$, with $\Delta x = 2x_{\max}/(2^n-1)$ and $j = 0,\dots,2^n-1$. Define the probabilities $p_j =  p_T(x_j)/C$, with the normalization
$C=  \sum_{j=0}^{2^n-1} p_T(x_j )$. This process is illustrated in Fig.~\ref{figure1}~(d). According to Ref.~\cite{Grover2002}, there exists a quantum algorithm $\mathcal G$ (which takes on the role of $\mathcal A$ in the previous section) such that we can prepare
\begin{equation} \label{eqGroverState}
\mathcal G \ket{0^n} =\sum_{j=0}^{2^n-1} \sqrt{p_j} \ket j.
\end{equation}
This algorithm runs in $\Ord{n}$ steps, provided that there is a method to efficiently sample the integrals
$\int_a^b p_T(x) dx$ for any $a,b$. These integrals can be efficiently sampled for any log-concave distribution, such as in the present case of the Gaussian distribution associated with $W_T$. 
We show the steps in the Appendix \ref{appendixGrover}.

Now consider the function $v(x):  \mathbbm R \to \mathbbm R$ which relates the Brownian motion to the option payoff. For the example of the European call option, the function is
\begin{equation}\label{eqBMtoPayoff}
v_{\rm euro}(x) = \max \{ 0, S_{0}e^{\sigma x + (r - \frac{1}{2} \sigma^2 )T} - K\}.
\end{equation}
At the discretization points of the Brownian motion we define
\begin{equation}
v(j) := v(x_j).
\end{equation}
One can find a binary approximation to this function over $n$ bits, i.e.~$\tilde v(j):  \{0,1\}^n \to \{0,1\}^n$, where we take 
the number of input bits to be the same as the number of output bits. The $n$ bits allow one to represent $2^n$ different floating points numbers, and with $n=n_1+n_2$ one can trade off the largest represented number $2^{n_1}$ with the accuracy of the representation, $2^{-n_2}$ \cite{IEEEFloat2008}. We can take $n=n_2$, by keeping track of the exponent of the floating point numbers offline. The accuracy for the function approximation is given by $\vert v(j) - \tilde v(j) \vert = \Ord{1/2^{n}}$, if $v(j)$ is sufficiently well behaved (e.g.~Lipschitz continuous.) The classical circuit depth of most such functions $\tilde v(j)$ discretizing real-world options payoffs is $\Ord{n}$. Via the reversible computing paradigm, this classical circuit can be turned into a reversible classical circuit using $\Ord{n}$ operations \cite{nielsen2002quantum}. Given the reversible classical circuit, the quantum circuit can be determined involving $\Ord{n}$ quantum gates. In other words we can implement the operation
\begin{equation}
\ket{j} \ket{0^n} \to \ket j \ket{\tilde v(j)}.
\end{equation}
See Appendix \ref{appendixArith} for more details on basic arithmetic operations and quantum circuits for options payoffs. 

We can now go through the steps of the quantum Monte Carlo algorithm. 
Sec.~\ref{sectionQuantumMC}  assumes availability of $\mathcal R$ with a real-valued function $v(x)$. Such a rotation can be directly implemented in some cases \cite{Low2016}. 
Here, we invoke the controlled rotation 
$\mathcal R \ket j \ket 0 = \ket j \left(\sqrt{1-\tilde v(x_j)} \ket 0 +\sqrt{\tilde v(x_j)} \ket 1 \right)$, with the discretized options payoff function $\tilde v(x)$.  See Appendix~\ref{App:Rot} for an implementation of this rotation by using an auxiliary register of qubits and the circuit for the options payoff. The steps are similar to before,
\begin{eqnarray}
\mathcal G \ket{0^n} &=& \sum_{j=0}^{2^n-1} \sqrt{p(x_j )} \ket j  \\
&\to&   \sum_{j=0}^{2^n-1} \sqrt{p(x_j )} \ket j \\&& \quad \left (\sqrt{1-\tilde v(x_j)} \ket 0 +\sqrt{\tilde v(x_j)} \ket 1 \right) =: \ket \chi. \nonumber
\end{eqnarray}
Measuring the ancilla in the state $\ket 1$, we obtain the expectation value
\begin{equation}
\mu = \bra \chi (\mathcal I_{2^{n}} \otimes \ket{1}\bra{1}) \ket \chi = \sum_{j=0}^{2^n-1} 
p_T(x_j ) \tilde v(x_j). 
\end{equation}
This expectation value $\mu$, assuming it can be measured exactly, determines the option price $\mathbbm E_{\mathbbm Q}[v(W_T)]$ to accuracy 
\begin{equation}
\vert \mu - \mathbbm E_{\mathbbm Q}[v(W_T)] \vert =: \nu ~.
\end{equation}
The error arises from the discretization of the probability density and the accuracy of the function approximation $\tilde v$. Using $n$ qubits the accuracy is given by $\nu = \Ord{ 2^{-n}}$.

We can employ phase estimation and Theorems \ref{theoremMean01}  and \ref{theoremMeanL2} to evaluate $\mu$ to a given accuracy and get a bound on the number of computational steps needed. For the European call option, 
we can show that the variance is bounded by $\mathbbm V_{\mathbbm Q}[f(S_T)] \leq \lambda^2$ where $\lambda^2 := \Ord{{\rm poly}(S_0, e^{rT},e^{\sigma^2 T},K)}$, see Appendix \ref{appendixBS}. Thus from Theorem \ref{theoremMeanL2} 
we know that we can use $\Ord{ \log(\lambda/\epsilon) \log \log (\lambda/\epsilon) } $ copies
of $\ket \chi$ and $\Ord{ (\lambda/\epsilon) \log^{3/2}(\lambda/\epsilon) \log \log (\lambda/\epsilon)}$ applications of $\mathcal U$ to provide an estimate $\hat \mu$ for $\mu$ up to additive error $\epsilon$ with success probability at least $2/3$.
The accuracy is $\epsilon < 4 \lambda$.
The total error  is
\begin{equation}
\vert \hat \mu - \mathbbm E_{\mathbbm Q}[f(S_T)] \vert \leq \epsilon + \nu,
\end{equation}
compounding the two sources from amplitude estimation and the discretization error. Discounting $\hat \mu$, see Eq.~(\ref{eqBSPriceEuro}), then retrieves an estimation of the option price $\hat \Pi$.
The total number of applications of $\mathcal U$ is 
\begin{equation}
\tOrd{\frac{\lambda}{\epsilon} },
\end{equation}
where $\tOrd{\cdot}$ suppresses polylogarithmic factors.
This quantity can be considered the analogue of the classical number of Monte Carlo runs
\cite{Montanaro2015,Xu2018}.
The required number of quantum steps is quadratically better than the classical number of steps in Eq.~(\ref{eqStepsClassical}), or the naive quantum case in Eq.~(\ref{eqStepsQuantumNaive}).

\section{Asian option pricing}
\label{sectionAsian}

Up to this point, we have discussed the quantum Monte Carlo pricing of derivatives via the illustrative example of the European call option. This call option can be priced analytically in the Black-Scholes-Merton framework, thus MC methods are in principle not required. 
Another family of options is the so-called Asian options, 
which depend on the average asset price before the maturity date \cite{Kemna1990}. 

\begin{defi} [Asian options]
The Asian call option payoff is defined as
\begin{equation}
f(A_T) = \max \{ 0, A_T - K\},
\end{equation}
where $K$ is the strike price and $T$ the maturity date. The arithmetic mean option value is defined via
\begin{equation}
A^{\rm arith}_T = \frac{1}{L} \sum_{l=1}^L S_{t_l},
\end{equation}
and the geometric mean option is defined via
\begin{equation}
A^{\rm geo}_T = \exp \frac{1}{L} \sum_{l=1}^L  \log S_{t_l},
\end{equation}
for pre-defined time points $0<t_1< \dots < t_L \leq T$, with $L\geq 1$.
\end{defi}
The following discussion assumes the BSM framework as before.
In this framework, the geometric mean Asian option can be priced analytically, while such a solution is not known for the arithmetic Asian option.  
Assume for this discussion 
that all adjacent time points are separated by the interval $\Delta t$, i.e.~$t_{l+1} - t_l = \Delta t = T/L$ for all $l=1,\dots,L-1$.
 Analogously to before, we can efficiently prepare via the Grover-Rudolph
algorithm \cite{Grover2002} a state that corresponds to the Gaussian normal distribution with
variance $\Delta t$
\begin{equation}
\ket{p_{\Delta t}} := \mathcal G \ket{0^m} = \sum_{j=0}^{2^m-1} \sqrt{p_{\Delta t}(x_j )} \ket j  
\end{equation}
This state uses $m$ qubits and takes $\Ord{m}$ steps to prepare. 
Then we prepare the product state with $L$ such states, i.e.,
\begin{equation}\label{eqAsianProductState}
\ket p := \ket{p_{\Delta t}} \dots  \ket{p_{\Delta t}}.
\end{equation}
This state uses $L m$ qubits and takes $\Ord{L m}$ steps to prepare. 

In addition, we require the operation
\begin{equation} \label{eqOpAsian}
\ket{j_1,\dots,j_L} \ket{0}=\ket{j_1,\dots,j_L} \ket{A(S_{t_1}(x_{j_1}),\dots,S_{t_L}(x_{j_L})}.
\end{equation}
Here, $A(S_{t_1}(x_{j_1}),\dots,S_{t_L}(x_{j_L}))$ is the average stock price corresponding to the Brownian path 
$x_{j_1},\dots,x_{j_L}$. This operation is easily computable. Each index $j$ is mapped to its corresponding point $x_j$
via $x_j = -x_{\max} + j \Delta x$ as before. Then start at the known $S_0$ and use
\begin{equation}
S_{t_{l+1}}(x) = S_{t_{l}} e^{\sigma x +(r-\sigma^2/2)\Delta t}
\end{equation}
to obtain the stock price at the next time point, where $x$ is a sample of the Brownian motion. This step can also be performed in the log domain \cite{Kemna1990}
\begin{equation}
\log S_{t_{l+1}}(x) = \log S_{t_{l}} + \sigma x +(r-\sigma^2/2)\Delta t.
\end{equation}
In this way, one obtains a state where the label $\ket{j_1,\dots,j_L}$ associated with the corresponding stock price path, 
\begin{equation}
\ket{j_1,\dots,j_L} \ket{S_{t_1}(x_{j_1})} \dots \ket{S_{t_L}(x_{j_L})}.
\end{equation}
Moreover, the average (both arithmetic and geometric) can be computed in a sequential manner, since we can
implement the step
\begin{eqnarray}
&\ket{j_1,\dots,j_L}  \ket{S_{t_l}(x_{j_l})} \ket{A(S_{t_1}(x_{j_1}),\dots,S_{t_l}(x_{j_l}))} \to \\
& 
\ket{j_1,\dots,j_L}  \ket{S_{t_{l+1}}(x_{j_{l+1}})} \ket{A(S_{t_1}(x_{j_1}),\dots,S_{t_{l+1}}(x_{j_{l+1}}))}. \nonumber
\end{eqnarray}
The steps are performed until the final time $t_L$ is reached and $A(S_{t_1}(x_{j_1}),\dots,S_{t_L}(x_{j_L}))$ is stored in a register of qubits.
Reversibility of the quantum arithmetic operations guarantees that registers storing the intermediate steps can be uncomputed. 
Applying operation Eq.~(\ref{eqOpAsian}) to the product state Eq.~(\ref{eqAsianProductState}) obtains
\begin{equation}
\sum_{j_1\dots j_L=0}^{2^m-1}\sqrt{ p_{j_1,\dots,j_L} } \ket {j_1, \dots, j_L}  \ket{A(S_{t_1}(x_{j_1}),\dots,S_{t_L}(x_{j_L}))} .
\end{equation}
with $\sqrt {p_{j_1,\dots,j_L}}:=\sqrt{p_{\Delta t}(x_{j_1} )} \dots \sqrt{p_{\Delta t}(x_{j_L} )} $.
Analogously to before, a conditional rotation of an ancilla qubit can be performed such that measuring the ancilla in the $\ket 1$ state obtains 
\begin{equation}
\sum_{j_1\dots j_L=0}^{2^m-1} p_{j_1,\dots,j_L} f(A(S_{t_1}(x_{j_1}),\dots,S_{t_L}(x_{j_L})) )\approx
\mathbbm E_{\mathbbm Q}[f(A)].
\end{equation}
The result is an approximation to the Black-Scholes price of the Asian option. The variance of the option can be bounded 
from the fact that the arithmetic mean upper bounds the geometric mean and the arithmetic mean itself is upper bounded by the expected maximum of the stock price $\max \{S_{t_1},\dots, S_{t_L} \}$. The variance of options such as calls or puts on the maximum of a stock in a time period can be bounded \cite{Shreve2004} similar to the variance bound of the European call option, which is presented in Appendix \ref{appendixBS}.
We thus obtain a similar speedup for the Asian options via Theorem \ref{theoremMeanL2}.

\section{Numerical Simulations}
While a practical quantum computer has yet to become a reality, we can exhibit the speedup of our quantum algorithm for options pricing numerically, and compare its performance with the
classical Monte Carlo method. Note that our quantum algorithm consists of two
main parts, see also Fig.~\ref{figure1}~(c). First, prepare the Brownian motion superposition (through $\mathcal{A}$) and encode the option payoff onto an ancilla qubit (through $\mathcal{R}$); and second, use amplitude amplification and phase estimation with repeated applications of $\mathcal{Q}$ to estimate the expectation value encoded in the ancilla qubit. 
The phase estimation subroutine can in principle be simulated using publicly available
quantum software packages such as Strawberry Fields \cite{Killoran2018} and ProjectQ
\cite{Steiger2018}.
However, to showcase the quadratic speed
up, we here perform phase estimation by using a single qubit rotated according to $e^{i\theta\sigma_z/2}$, where $\theta$ is the predetermined phase.

We perform numerics for a phase $\theta$ given by the European
call option, see Sec.~\ref{sectionBS}, which can be priced analytically. 
The analytical price $\Pi$ is  computed directly from the Black-Scholes-Merton formula, see Result \ref{resultBS}.
We provide an estimate $\hat \theta$ using both
quantum phase estimation and the standard classical Monte Carlo method.
Here, the analytical price $\Pi$ is used both as an input to the single-qubit phase estimation 
via $\theta$ from Eq.~(\ref{eqMeanToAngle}), as well as a benchmark for the resultant simulations.
The single-qubit phase estimation is described in Appendix \ref{appendixPhaseEstimation}.
We define the
corresponding estimation error as the difference between the estimated price
$\hat \Pi$ and analytical price $\Pi$, i.e.
\begin{equation}
 \rm Error := |\hat \Pi- \Pi|.
\label{error}
\end{equation}
In the figures and following discussion, we will use subscripts $Q$ and $C$ to
denote the quantum and classical estimations respectively. The estimation error
follows a power-law behavior with the number of MC steps $k$ undertaken
\begin{equation}
{\rm Error} = a~k^\zeta,
\label{powerlaw}
\end{equation}
where $\zeta$ is the scaling exponent and $a$ is a constant. As discussed in Sec.~\ref{sectionMC}, as well as
obtained in our simulations, the scaling exponent for classical MC estimation
is $\zeta_{\rm C}=-1/2$. 
For the quantum case, the $k_Q$ is the total number of applications of the single-qubit unitary.
The simulations are performed such that quantum and classical estimates have a similar confidence ($>99.5\%$), which implies a number of independent single-qubit phase estimation runs of $D\approx 24$ \cite{Xu2018}, see also Appendix \ref{appendixPhaseEstimation}.

\begin{figure}[t!]
\includegraphics[width=0.5\textwidth]{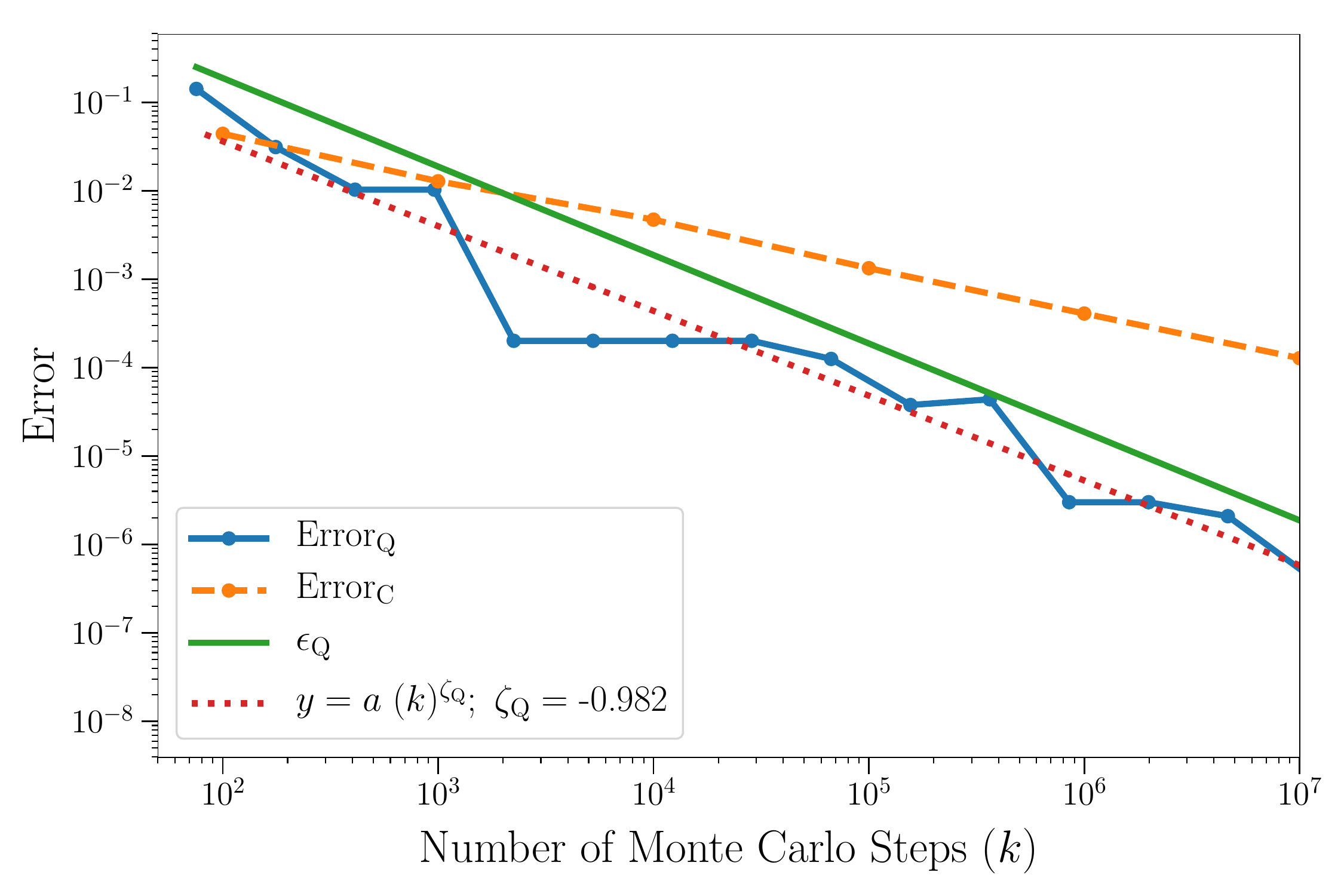}
\vspace{-0.6cm}\caption{Scaling of the error in classical and quantum MC methods (defined in Eq.\ (\ref{error})) plotted
against number of MC steps for a European call option in log-log scale with
$S_0=\$100, ~K=\$50, ~r=0.05, ~\sigma=0.2$, $T=1$, and $D=24$.
Subscripts $C$ and $Q$ denote the errors from classical MC and
quantum phase estimation, respectively. Evidently, the error for the quantum algorithm (with a fitted slope of
$\zeta_{\rm Q}=-0.982$) scales almost quadratically faster than the classical
MC method (which has $\zeta_{\rm C}=-0.5$). The theoretical upper bound on the
error in quantum algorithm is shown by the solid green curve, which corresponds to
$\zeta_{\rm Q}=-1$.}
	\label{fig2}
\end{figure}

Fig.~\ref{fig2} shows a comparison between the error scalings for our quantum
algorithm (blue solid curve with markers) and classical MC (orange dashed
curve). The parameters for this figure are: $S_0=\$100, ~K=\$50, ~r=0.05, ~\sigma=0.2$, $T=1$, and $D=24$. The analytical price was determined to be $\Pi=\$10.5$ and rescaled to the corresponding $\theta = 2\arccos(1-2 \Pi / S_{0})$ via Eq.~\eqref{eqMeanToAngle} for use in the phase estimation. In addition, we have also plotted a fit of the quantum error to the 
power law given in Eq.~(\ref{powerlaw}) resulting in the scaling exponent to be
$\zeta_{\rm Q}=-0.982$. It is evident that the scaling exponent of quantum phase estimation is almost twice the classical one. Indeed, the green solid curve shows the upper bound on the errors in quantum estimation
defined as \cite{Xu2018}, see also Appendix~\ref{appendixPhaseEstimation}:
\begin{equation}\label{UpperBound}
\epsilon_{\rm Q} := \left|\cos\left(\frac{\hat \theta}{2} 
                             +\frac{\pi}{k_Q }\right)
                             -\cos\left(\frac{\hat\theta}{2}\right)\right|~,
\end{equation}
where $\hat\theta$ is the phase estimated via the quantum algorithm. This upper bound has a scaling exponent of $-1$ and hence straightforwardly demonstrates the quadratic speedup in the number of steps for phase estimation to a given error.

To test the robustness of the quadratic speedup in scaling, we vary the strike
price and plot the ratio of the quantum to the classical scaling exponents
$\zeta_{\rm Q}/\zeta_{\rm C}$ in Fig.~\ref{fig3}. We obtain an almost quadratic
advantage in estimation overhead for all strike prices. Similar tests by varying other parameters also show a robust quadratic quantum speedup of the Monte Carlo estimation. 

\section{Discussion and conclusion}

In this work, we have presented a quantum algorithm for the pricing of financial derivatives. We have assumed that the distribution of the underlying random variables, i.e.~the martingale measure, is known and the corresponding quantum states can be prepared efficiently. In addition, we assume efficient computability of the derivative payoff function. 
Under these assumptions, we exhibit a quadratic speedup in the number of samples required to estimate the price of the derivative up to a given error: if the desired accuracy is $\epsilon$, then classical methods show a $1/\epsilon^2$ dependency in the number of samples, while the quantum algorithm shows a
$1/\epsilon$ dependency.

\begin{figure}[t!]
\vspace{-0.5cm}
\includegraphics[width=0.525\textwidth]{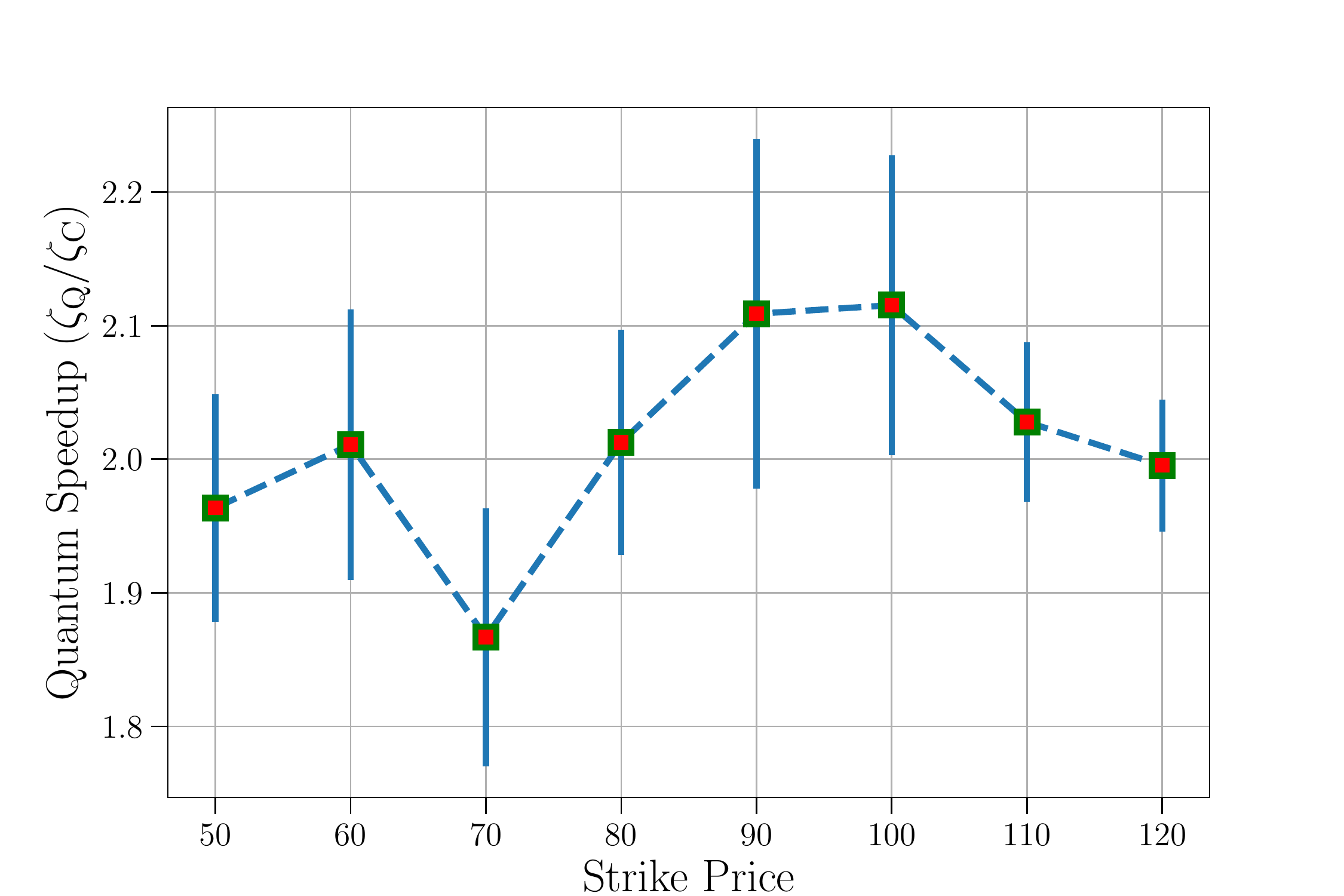}
	\caption{Ratio of the quantum to classical scaling exponents. The quantum scaling is obtained by
fitting the simulations results to the power law in Eq.~(\ref{powerlaw}), while the classical
scaling exponent is taken to be $-0.5$. Results are plotted with varying strike price $K$ (in dollars) and fixing
other parameters to be the same as in Fig.~\ref{fig2}. An almost quadratic speed up is obtained for all chosen values of $K$.  }
	\label{fig3}
\end{figure}

As an exemplary case, we have discussed  European call options for which analytical solutions are known. In addition, we have discussed Asian options, which in the arithmetic averaging case require Monte Carlo methods to be priced. Our approach can in principle be applied to any derivative which has a payoff that is a function that can be efficiently broken down into elementary arithmetic operations within a quantum circuit, and which can also depend on an average over multiple time windows. Future work will extend these discussions to the complex payoff functions often used at leading financial institutions, as well as addressing more complicated stochastic models. 

Monte Carlo simulations play a major role in managing the risk that a financial institution is exposed to \cite{Follmer2004}. Especially after the financial crisis of 2008-9, sophisticated risk management is increasingly important to banks internally and also required by government regulators \cite{Green2015,Zeitsch2017}. Such risk analysis falls under the umbrella of so-called valuation adjustments (VA), or XVA \cite{Green2015,Zeitsch2017}, where X stands for the type of risk under consideration. An example is CVA, where the counterparty credit risk is modeled. Such a valuation adjusts the price of the derivative based on the risk that the counterparty in that financial contract runs out of money.

XVA calculations are a major computational effort for groups (`desks') at financial institutions that handle complex derivatives such as those based on interest rates. 
For complex financial derivatives, such risk management involves a large amount of Monte Carlo simulations. Different Monte Carlo runs assess the price of a derivative under various scenarios. 
Determining the risk of the complete portfolio of a desk often requires overnight calculations of the prices according to various risk scenarios. Quantum computers promise a significant speedup for such computations. In principle, overnight calculations could be reduced to much shorter time scales (such as minutes), which would allow a more real time analysis of  risk. Such close-to real time analysis would allow the institution to react faster to changing market conditions and to profit from trading opportunities. 

The qubit model for universal quantum computing was employed here to provide a succinct discussion. However, one is not restricted to such a setting but can use other universal quantum computational models such adiabatic quantum computation or continuous variable (CV) quantum computation. In the continuous variable setting, instead of qubits one has oscillators which in principle have infinite dimensional Hilbert spaces. In this setting, preparing Gaussian states for the probability distributions can be done in a straightforward manner, instead of employing the relatively complicated preparation routine according to Grover-Rudolph. Investigating the promising advantages of the CV setting in a financial context in more detail will be left for future work. 

\acknowledgements
We acknowledge Juan Miguel Arrazola, John C. Hull, Ali Najmaie and Vikash Ranjan for insightful discussions. 


\appendix 

\section{Black-Scholes calculations}
\label{appendixBS}

Here, we show several calculations related to the Black-Scholes-Merton model. We show the solution to the stochastic differential equation, the martingale property of the stock price, the BSM price for the European call option, and an upper bound for the variance of the BSM price for the call option.
First note the following result.
\begin{result} (informal) \label{resultQuadraticVariation}
For the infinitesimal Brownian increment it holds that 
\begin{equation}
dW_t^2 = dt
\end{equation}
\end{result}
This result can be reasoned by the variance of the Brownian increment being proportional to the 
time interval of the increment. 
Next, we show the solution to the following stochastic differential equation.
\begin{lemma}
The stochastic differential equation 
\begin{equation}
dS_t = S_t \alpha dt +  S_t \sigma dW_t,
\end{equation}
is solved by 
\begin{equation}
S_t = S_0 e^{\sigma W_t + (\alpha - \frac{\sigma^2}{2}) t}.
\end{equation}
\end{lemma}
\begin{proof}
We show this result by finding the differential $dS_t$ given $S_t$. The quantity $S_t$ can be seen as a function 
$s(t,x)$.
To obtain the differential, expand this function to second order in $dx$ and use the solution $S_t$ to obtain
\begin{eqnarray}
d s(t,x) &=& \frac{\partial s}{\partial t} dt+  \frac{\partial s}{\partial x} dx +\frac{1}{2} \frac{\partial^2 s}{\partial x^2} dx^2 +\Ord{dt^2,dx^3}\\ \nonumber
&=& (\alpha - \frac{\sigma^2}{2}) s(t,x) dt + \sigma s(t,x) dx + \frac{\sigma^2}{2} s(t,x) dx^2.
\end{eqnarray}
Using $dx = dW_t$ and $dx^2 = dW_t^2 = dt$ from Result \ref{resultQuadraticVariation} leads to
the differential
\begin{equation}
dS_t = S_t \alpha dt +  S_t \sigma dW_t.
\end{equation} 
\end{proof}
Next, we show the martingale property of the discounted stock price.
\begin{lemma}\label{lemmaMartingale}
Under the $\mathbbm Q$ measure, the stock price is a martingale, i.e. 
\begin{equation}
S_0=e^{-r T} \mathbbm E_{\mathbbm Q} [S_T].
\end{equation}
\end{lemma}
\begin{proof}
\begin{eqnarray}
e^{-r T} \mathbbm E_{\mathbbm Q} [S_T] &=&S_0 \mathbbm E_{\mathbbm Q} [ e^{\sigma \tilde W_T - \frac{\sigma^2 T}{ 2} }] \\
&=& e^{-\frac{\sigma^2T}{2} }   \frac{S_0 }{\sqrt{2\pi T}}  \int_{-\infty}^\infty dx\ e^{-\frac{x^2}{2T}} e^{\sigma x}\\
&=&S_0.
\end{eqnarray} 
In the last step we simplified by substitution and completed the square as
\begin{eqnarray}
\int_{-\infty}^\infty dx\ e^{-\frac{x^2}{2T}} e^{\sigma x} &=& \sqrt{T} \int_{-\infty}^\infty dx\ e^{-\frac{x^2}{2}} e^{\sigma \sqrt{T} x}\\
&=& \sqrt{T} e^{\frac{\sigma^2 T}{2}} \int_{-\infty}^\infty dx\ e^{-\frac{(x-\sqrt{T}\sigma)^2}{2}} \\
&=&  \sqrt{2\pi T} e^{\frac{\sigma^2 T}{2}}.
\end{eqnarray}
\end{proof}
Via a similar calculation we can obtain the price for the call option. 
\begin{lemma}
The Black-Scholes price for the European call option is given by
\begin{equation}
\Pi = S_0 \Phi (d_1) - K e^{-rT} \Phi (d_2).
\end{equation}
\end{lemma}
\begin{proof}
To compute the price we start at
\begin{equation}
\Pi =  e^{-r T} \mathbbm E_{\mathbbm Q} [f(S_T)].
\end{equation}
Note that for the call option we can write 
\begin{equation}
f(S_T) = \max \{ 0, S_T - K\} = (S_T-K) \mathbbm 1_{S_T \geq K},
\end{equation}
where $ \mathbbm 1_{x}$ is the indicator function on the set $x$.
Using the solution for $S_T$ we have
\begin{equation}
f(\tilde W_T)  = (S_0 e^{\sigma \tilde W_T +(r -\sigma^2/2) T}-K) \mathbbm 1_{\tilde W_T \geq\frac{ \left(\log \frac{K}{S_0} -(r-\frac{\sigma^2}{2}) T \right)}{\sigma}}.
\end{equation}
Compute the part involving the strike price $K$
\begin{eqnarray}
\mathbbm E \left[\mathbbm 1_{\tilde W_T \geq\frac{ \left(\log \frac{K}{S_0} -(r-\frac{\sigma^2}{2}) T \right)}{\sigma}}\right] &=&  
\int_{\frac{ \left(\log \frac{K}{S_0} -(r-\frac{\sigma^2}{2}) T \right)}{\sigma \sqrt{T}}}^\infty dx\ p(x) \nonumber \\
&=&\Phi (d_2).
\end{eqnarray}
Here we use the cumulative distribution function 
\begin{equation}
\Phi(x) = \int_{-\infty}^x dy\ p(y) := \frac{1}{\sqrt {2\pi}} \int_{-\infty}^x dy\ e^{-\frac{y^2}{2}}.
\end{equation}
The part involving the stock price is
\begin{eqnarray}
&\mathbbm E \left[S_0 e^{\sigma \tilde W_T -\sigma^2 T/2} \mathbbm 1_{\tilde W_T \geq\frac{ \left(\log \frac{K}{S_0} -(r-\frac{\sigma^2}{2}) T \right)}{\sigma}}\right]  \\
&= S_0e^{-\sigma^2 T/2} \int_{\frac{ \left(\log \frac{K}{S_0} -(r-\frac{\sigma^2}{2}) T \right)}{\sigma \sqrt{T}}}^\infty dx\ p(x) e^{ \sigma \sqrt{T} x}  \nonumber\\
&= S_0 \Phi (d_1). \nonumber
\end{eqnarray}
Here, 
\begin{eqnarray}
d_1&=& \frac{1}{\sigma \sqrt{T}} \left[\log \left( \frac{S_0}{K} \right) + \left(r+\frac{\sigma^2}{2} \right) T \right], \\
d_2&=& \frac{1}{\sigma \sqrt{T}} \left[\log \left( \frac{S_0}{K} \right) + \left(r-\frac{\sigma^2}{2} \right) T \right],
\end{eqnarray}
This leads to the Black-Scholes price for the call option.
\end{proof}

\begin{lemma}
The variance of the European call option $f(S_T)$ under the risk-neutral probability measure $\mathbbm Q$ can be bounded as $\mathbbm V_{\mathbbm Q}[f(S_T)] \leq \lambda^2$ with $\lambda^2 := \Ord{{\rm poly}(S_0, e^{rT},e^{\sigma^2 T},K)}$. 
\end{lemma}
\begin{proof}
The variance is exactly computable. First note that $S_T^2 = S_0^2 e^{2\sigma \tilde W_T +(2r -\sigma^2) T}$. We have
\begin{eqnarray}
\mathbbm E_{\mathbbm Q} [f(S_T)^2] &=&  \mathbbm E_{\mathbbm Q}  [(S_T-K)^2 \mathbbm 1_{S_T \geq K} ] \\
&=&   \mathbbm E_{\mathbbm Q}  [(S_T^2-2 S_T K + K^2) \mathbbm 1_{S_T \geq K} ] \nonumber .
\end{eqnarray}
The last two terms were calculated analogously in the Black-Scholes price already. They are
\begin{eqnarray}
2 K E_{\mathbbm Q}  [ S_T\mathbbm 1_{S_T \geq K}  ] &=& 2 K e^{rT} S_0 \Phi (d_1)\\
E_{\mathbbm Q}  [K^2\mathbbm 1_{S_T \geq K} ] &=& K^2 \Phi (d_2).
\end{eqnarray}
The first term is $\mathbbm E_{\mathbbm Q} [S_T^2\mathbbm 1_{S_T \geq K} ]$ which is proportional to 
\begin{eqnarray}
E_{\mathbbm Q}  \left[ e^{2\sigma \tilde W_T}  \mathbbm 1_{\tilde W_T \geq\frac{ \left(\log \frac{K}{S_0} -(r-\frac{\sigma^2}{2}) T \right)}{\sigma}}\right] \\
\quad =\int_{\frac{ \left(\log \frac{K}{S_0} -(r-\frac{\sigma^2}{2}) T \right)}{\sigma \sqrt{T}}}^\infty dx\ p(x) e^{ 2 \sigma \sqrt{T} x}\\
\quad = e^{2\sigma^2 T} \Phi(d_3).
\end{eqnarray}
with
\begin{equation}
d_3= \frac{1}{\sigma \sqrt{T}} \left[\log \left( \frac{S_0}{K} \right) + \left(r+\frac{3\sigma^2}{2} \right) T \right].
\end{equation}
Putting it all together we obtain
\begin{eqnarray}
V_{\mathbbm Q}[f(S_T)]  &=& E_{\mathbbm Q} [f(S_T)^2] -E_{\mathbbm Q} [f(S_T)]^2 \\
&=& e^{(2r +\sigma^2) T} S_0^2 \Phi(d_3) -2 K e^{rT} S_0 \Phi (d_1)\nonumber \\ && + K^2 \Phi (d_2)
 - (S_0 e^{rT} \Phi (d_1) - K \Phi (d_2) )^2 \nonumber \\ 
&=& \Ord{ S_0^2 e^{(2r +\sigma^2) T} + S_0 K e^{rT}+ K^2 }.
\end{eqnarray}
We obtain an upper bound as
$\mathbbm V_{\mathbbm Q}[f(S_T)] = \Ord{{\rm poly}(S_0, e^{rT},e^{\sigma^2 T},K)}$.
\end{proof}

\section{Introduction to quantum computing}\label{appendixQIntro}

Instead of the elementary bits of conventional computing, which take either the value $0$ or $1$, the unit of information in quantum computing is known as a \emph{qubit}. A qubit is not restricted to being in one of the two states $0$ or $1$ exclusively, and can instead exist in a superposition. In quantum mechanics we use the \emph{bra-ket} notation, with the kets $\ket{0}$ and $\ket{1}$ representing $0$ and $1$, respectively. The state of a qubit can then be written as
\begin{equation}
\ket{\psi} = \alpha \ket{0} + \beta \ket{1},
\end{equation}
in terms of so-called \emph{amplitudes} $\alpha, \beta \in \mathbb{C}$. This means that the qubit can be in either of the states with some probability. Measuring the qubit collapses it onto $\ket{0}$ with probability $|\alpha|^{2}$ and onto $\ket{1}$ with probability $|\beta|^{2}$. Hence, $|\alpha|^{2} + |\beta|^{2} = 1$.

States of $n$ qubits exist in a $2^{n}$-dimensional Hilbert space and can be described by the vector
\begin{equation}
\ket{\psi} = \sum_{i = 0}^{2^{n} - 1} \alpha_{i}\ket{i},
\end{equation}
with $\alpha_{i} \in \mathbb{C}$ and $\ket{i} = \ket{i_{1}}\otimes \ket{i_{2}}\otimes \ldots \otimes \ket{i_{n}}$ the vector of basis states over $n$ qubits such that $[i_{1},i_{2}, \ldots , i_{n}]$ is the length $n$ binary representation of $i$. The probability that the $n$ qubits are in state $\ket{i}$ is equal to $|\alpha_{i}|^{2}$, where $\sum_{i=0}^{2^{n}-1}|\alpha_{i}|^{2} = 1$. Note that the tensor product symbol $\otimes$ denotes the joining of two quantum systems (e.g. qubits) in quantum mechanics.

Isolated quantum system evolve according to unitary transformations, i.e.~so that $\ket{\psi} $ goes to $\ket{\psi '} = U \ket{\psi}$ for some unitary $U$ (which satisfies $U^{\dagger}U = UU^{\dagger} = \mathcal{I}$, where $\dagger$ is the conjugate transpose and $\mathcal{I}$ is the identity operator). Unitaries can be controlled externally and also applied to collections of multiple qubits, as is often the case in this work. A standard library of one- and two-qubit unitaries can also be applied, and are often called quantum gates~\cite{nielsen2002quantum} in analogy to the gates used used for binary logic operations, e.g.~AND, OR, etc. This work uses the Hadamard gate $\mathcal{H}$, which acts on one qubit with the transformation
\begin{equation}
\ket{0} \rightarrow \frac{1}{\sqrt{2}}\left(\ket{0}+\ket{1} \right), \qquad \ket{1} \rightarrow \frac{1}{\sqrt{2}}\left(\ket{0}-\ket{1} \right).
\end{equation}
The inverse quantum Fourier transform $QFT^{-1}$ is also utilized in this work, which acts to perform an inverse discrete Fourier transform on the amplitudes $\alpha_{j}$, i.e. taking the corresponding state $\ket{j}$ to
\begin{equation}
\frac{1}{\sqrt{2^{n}}}\sum_{k=0}^{2^{n} - 1} \omega^{jk} \ket{k},
\end{equation}
with $\omega = e^{- 2 \pi i / 2^{n}}$. Unitaries acting on one or more qubits can also be controlled by another qubit, e.g.~so that the unitary is enacted if the control qubit is in state $\ket{1}$ and the unitary is not enacted if the control qubit is in state $\ket{0}$. One important example is the controlled-NOT (CNOT) gate, which swaps the state of a qubit from $\ket{0}$ to $\ket{1}$ (and vice versa) whenever a control qubit is in state $\ket{1}$. Its extension is the Toffoli gate, which swaps the state of a qubit from $\ket{0}$ to $\ket{1}$ (and vice versa) only when \emph{two} control qubits are in state $\ket{1}$.

Compositions of unitaries acting on multiple qubits, prepared in various initial states, can be created to form quantum circuits. These quantum circuits are able to carry out quantum algorithms that may perform a task faster than on a classical device, e.g.~Shor's algorithm ~\cite{nielsen2002quantum}. The unitary nature of quantum algorithms means that all quantum circuits are \emph{reversible}. Reversibility has implications for the structure of quantum circuits, and marks a departure from the more familiar classical circuits. For example, the AND gate is not reversible because a single bit is returned from which the two input bits cannot be inferred in all cases.

Quantum circuits can be represented pictorially using a quantum circuit diagram, as is the case in Fig.~\ref{figure1} (a) and (c). Here, each qubit is represented by a horizontal line, with time moving from left to right. The unitaries are applied to various qubits by drawing a box over the qubit lines, with control from another qubit symbolized by a black dot and line connecting to the box. Measurements on the qubits are represented by the $\Qcircuit @C=1em @R=.3em {& \meter & \qw}$ symbol. Finally, the CNOT gate is written as $\Qcircuit @C=1em @R=.3em {& \ctrl{1} & \qw \\ & \targ & \qw \\}$, with the upper qubit acting as the control and the lower qubit acting as the target, while the Toffoli gate is represented as
$\Qcircuit @C=1em @R=.3em {& \ctrl{1} & \qw \\& \ctrl{1} & \qw \\ & \targ & \qw \\}$, with control from the top two qubits and target on the bottom qubit.

\section{Quantum circuits for arithmetic operations}
\label{appendixArith}

In this Appendix, we review how integers and real numbers are represented and processed 
on a quantum computer \cite{nielsen2002quantum}.
An integer $0\leq a < N = 2^{m}$ can be represented in binary with $m$ bits $x_i$, $i=0,\dots,m-1$ such that 
\begin{equation} 
a =2^0 x_0 + 2^1 x_1 + 2^2 x_2 + \dots + 2^{m-1} x_{m-1}.
\end{equation}
The largest number is $N-1$.
The qubit encoding of the integer is where $\ket{a}$ refers to an $m$ qubit register prepared in the computational basis state given by the $x_i$, i.e.~$\ket{a} = \ket{x_0,x_1,\dots,x_m}$.
For real numbers $0 \leq r < 1$, we can use $m$ bits $b_i$, $i=0,\dots,m-1$ such that
\begin{equation}
r=  \frac{b_0}{2} + \frac{b_1}{4} + \dots + \frac{b_{m-1}}{2^{m}} =: [.b_1,b_2,\dots,b_{m-1}].
\end{equation}
The accuracy is $1/2^m$. 
The qubit encoding of the real number is where $\ket{r}$ refers to a $m$ qubit register prepared in the computational basis state given by the $b_i$, i.e.~$\ket{r} = \ket{b_0,b_1,\dots,b_{m-1}}$.
For signed integers or reals, we have an additional sign quantum bit $\ket s$.

Any operation performed on a classical computer can be written as a transformation 
from $n$ to $m$ bits, that is $F: \{ 0,1\}^n \to  \{ 0,1\}^m$. A central result of reversible computing is
that the number of input and output bits can be made the same and the function $F$ can be mapped to a function
$F': \{ 0,1\}^{n+m} \to  \{ 0,1\}^{n+m}$ which is given by $F'(x,y) = (x,y \oplus F(x))$. $F'$ is a reversible function and a permutation. This permutation can be realized with a circuit that consist only of negation and Toffoli gates. If $F$ is efficiently computable then the circuit depth is at most a polynomial in $n+m$. The classical circuit then immediately translates into a quantum circuit consisting of bit flip $\sigma_z$ operations and Toffoli gates. The Toffoli gate can be broken down into a series of two qubit CNOTs and Hadamard and T gates. 

We now review how basic arithmetic operations can be performed on a quantum computer. 
Using the following basic gates, a number of arithmetic operations can be constructed \cite{Vedral1996}:
\begin{equation}
\mbox{
\Qcircuit @C=.5em @R=0em @!R {
 \lstick{1} & \multigate{2}{SUM} & \qw & & & \lstick{1} & \qw & \qw& \ctrl{2} & \qw \\
\lstick{2} & \ghost{SUM} & \qw &
\push{\rule{.3em}{0em}=\rule{.3em}{0em}} & & \lstick{2} & \ctrl{1} &\qw & \qw & \qw \\
 \lstick{3} & \ghost{SUM} & \qw & & & \lstick{3} & \targ  &\qw & \targ  & \qw
}
}
\end{equation}

\begin{equation}
\mbox{
\Qcircuit @C=.5em @R=0em @!R {
 \lstick{1} & \multigate{3}{CY} & \qw & & & \lstick{1} & \qw & \qw& \qw & \qw & \ctrl{2} & \qw \\
\lstick{2} & \ghost{CY} & \qw &
\push{\rule{.3em}{0em}=\rule{.3em}{0em}} & & \lstick{2} & \ctrl{1} &\qw & \ctrl{1} & \qw & \qw & \qw \\
 \lstick{3} & \ghost{CY} & \qw & & & \lstick{3} & \ctrl{1}  &\qw & \targ & \qw & \ctrl{1}  & \qw
 \\
 \lstick{4} & \ghost{CY} & \qw & & & \lstick{4} & \targ  &\qw & \qw & \qw & \targ  & \qw
}
}
\end{equation}
where $CY$ represents the `carry' operation. Composing these gates can achieve addition, multiplication, exponentiation and other operations. 
Numerous other works provide circuits for these operations with improved performance and lower gate requirements \cite{Beckman1996,VanMeter2005,Draper2006,Takahashi2010,Bhaskar2016,Wiebe2016arith}.

There exists a quantum circuit that performs the addition modulo N. Given two integers $0\leq a,b < N $ we can  construct a circuit of the form 
\begin{equation}
\mbox{
\Qcircuit @C=.5em @R=0em @!R {
 \lstick{\ket{a} } & \multigate{2}{ADD} & \qw & \rstick{\ket{a}} \\
 & &  &  \\
 \lstick{\ket{b} } & \ghost{ADD} & \qw &  \rstick{\ket{a+b},}
}
}
\end{equation}
defined via a gate sequence of SUM and CYs in Ref. \cite{Vedral1996}.
In addition, a circuit for addition modulo $N$ can be constructed 
\begin{equation}
\mbox{
\Qcircuit @C=.5em @R=0em @!R {
 \lstick{\ket{a} } & \multigate{2}{ADD\ N} & \qw & \rstick{\ket{a}} \\
 & &  &  \\
 \lstick{\ket{b} } & \ghost{ADD\ N} & \qw &  \rstick{\ket{a+b \mod N}.}
}
}
\end{equation}
There also exists a quantum circuit that performs the multiplication
\begin{equation}
\mbox{
\Qcircuit @C=.5em @R=0em @!R {
 \lstick{\ket{x} } & \multigate{2}{MULT(a)\ N} & \qw & \rstick{\ket{x}} \\
 & &  &  \\
 \lstick{\ket{0} } & \ghost{MULT(a) \ N} & \qw &  \rstick{\ket{a \times x \mod N},}
}
}
\end{equation}
as well as a quantum circuit that performs the exponentiation
\begin{equation}
\mbox{
\Qcircuit @C=.5em @R=0em @!R {
 \lstick{\ket{x} } & \multigate{2}{EXP(a)} & \qw & \rstick{\ket{x}} \\
 & &  &  \\
 \lstick{\ket{0} } & \ghost{EXP(a)} & \qw &  \rstick{\ket{a^x \mod N}.}
}
}
\end{equation}
We would like to implement the call option payoff function
\begin{equation}
a^+ = \max \{ 0, a \}.
\end{equation}
We can implement this as a reversible circuit
\begin{equation}
\mbox{
\Qcircuit @C=.5em @R=0em @!R {
 \lstick{\ket{a,s} } & \multigate{2}{MAX(0)} & \qw & \rstick{\ket{a,s}} \\
 & &  &  \\
 \lstick{\ket{0} } & \ghost{MAX(0)} & \qw &  \rstick{\ket{ a^+ },}
}
}
\end{equation}
which performs
\begin{equation}
\ket{a,s,0} \to
\begin{cases}
\ket {a,s,a}& \text{if } \ket s =\ket 0\\
\ket {a,s,0}              & \text{if } \ket s =\ket 1
\end{cases}.
\end{equation}
Here, the sign bit is used as a controller and a controlled addition is performed if the sign bit is positive.
There exists a circuit for mapping the Brownian motion to the stock price, which consists of the basic operations shown above, 
\begin{equation}
\mbox{
\Qcircuit @C=.5em @R=0em @!R {
 \lstick{\ket{x} } & \multigate{2}{S(\sigma, r, t)} & \qw & \rstick{\ket{x}} \\
 & &  &  \\
 \lstick{\ket{0} } & \ghost{S(\sigma, r, t)} & \qw &  \rstick{\ket{ e^{\sigma x +(r-\sigma^2/2)t }}.
}
}
}
\end{equation}
Combining the stock price with the max function, we obtain the circuit for mapping the Brownian motion outcome to the payoff for the European call option
\begin{align} \label{circPayoff}
 \Qcircuit @C=.5em @R=0em @!R {
 \lstick{\ket{x} } & \multigate{2}{CALL(K, \sigma,r, T)} & \qw & \rstick{\ket{x}} \\
 & &  &  \\
 \lstick{\ket{0} } & \ghost{CALL(K,\sigma,r, T)} & \qw &  \rstick{\ket{\tilde v_{\rm euro}(x) }.
}
}
\end{align}
with $\tilde v_{\rm euro}(x) \equiv \tilde v_{\rm euro}(x,K,\sigma,r, t)$ the bit approximation of the payoff function Eq.~(\ref{eqBMtoPayoff}).

\section{Applying the operator $\mathcal R$}\label{App:Rot}
This Appendix shows implementation of the operator $\mathcal R$, which rotates an ancilla based on the payoff function, where the payoff function is given by an $n$-bit quantum circuit, similar to circuit (\ref{circPayoff}).
Using the option payoff circuit, the steps to implement $\mathcal R$ are (the notation omits ancilla qubits in the $ \ket 0 $ state)
\begin{eqnarray}
 \ket j   &\to & \ket j \ket {\tilde v(x_j)} \\
&\to &\ket j \ket {\tilde v(x_j)}  \left (\sqrt{1-\tilde v(x_j)} \ket 0 + \sqrt{\tilde v(x_j)} \ket 1 \right) \\
 &\to & \ket j \left (\sqrt{1-\tilde v(x_j)} \ket 0 + \sqrt{\tilde v(x_j)} \ket 1 \right ) \\
 &\equiv& \mathcal R \ket j \ket 0.
\end{eqnarray}

\section{Preparation of the quantum state encoding the sampling distribution}
\label{appendixGrover}

Consider a probability density $p(x)$ of a single variable $x$. Assume we have discretized the probability over some interval,
such that for some integer $n$, we have $\{ p_j \}$ for $j=0,\dots, 2^n-1$.
Assume that $\sum_j p_j =1$.
The task is to show an algorithm $\mathcal G$ without measurements such that 
\begin{equation}
\mathcal G \ket{0^n} =: \ket \psi = \sum_{j=0}^{2^n-1} \sqrt{p_j} \ket j.
\end{equation}
Further assume that there exists a shallow classical circuit that can efficiently compute the sums (subnorms) 
\begin{equation}
\sum_{j=a}^b p_j \approx \int_{x_a}^{x_b} p(x) dx,
\end{equation}
for any $a\leq b = 0,\dots, 2^n-1$, $a \leq b$.
Thus for $m=1,\dots,n$ we can efficiently compute the probabilities 
\begin{equation}
p_k^{(m)} = \sum_{j= k 2^{n-m}}^{(k+1) 2^{n-m}-1} p_j,
\end{equation}
with $k=0,\dots,2^m-1$.

The quantum algorithm goes as follows \cite{Grover2002}. For $m<n$, assume we have prepared the state
\begin{equation}
\ket{\psi^{(m)}} = \sum_{k=0}^{2^m-1} \sqrt{p_k^{(m)}} \ket k.
\end{equation}
We would like to show by induction that we can prepare the state
\begin{equation}
\ket{\psi^{(m+1)}} = \sum_{k=0}^{2^{m+1}-1} \sqrt{p_k^{(m+1)}} \ket k.
\end{equation}
Define the quantities 
\begin{equation}
f(k,m) = \frac{p_{2k}^{(m+1)} }{p_k^{(m)}}, 
\end{equation}
where in the denominator there is the sum of all the elements of the $k$-th interval at the $m$-th discretization level and in the numerator we have the sum of the left half of these elements. This quantity allows to go up one level of discretization to $m+1$. 
Also define $\theta_k^{(m)} = \arccos \sqrt{f(k,m)}$. 
The operation
\begin{equation}
\ket k \ket 0 \to \ket k \ket  {\theta_k^{(m)}}
\end{equation}
is enabled by the efficient computability of $f(k,m)$.
Now proceed 
\begin{eqnarray}
\ket{\psi^{(m)}} \ket 0 \ket 0 &\to&\sum_{k=0}^{2^m-1} \sqrt{p_k^{(m)}} \ket k \ket{ \theta_k^{(m)}} \ket 0  \\
&\to & \sum_{k=0}^{2^m-1} \sqrt{p_k^{(m)}} \ket k \left ( \cos \theta_k^{(m)} \ket 0 \right. \\ \nonumber & & \qquad \qquad \qquad \left. + \sin \theta_k^{(m)} \ket 1\right) \nonumber \\
&\equiv & \ket{\psi^{(m+1)}} .
\end{eqnarray}
In the second step the register $\ket {\theta_k^{(m)}}$ was uncomputed. 

\section{Phase estimation}
\label{appendixPhaseEstimation}

This Appendix shows the basic steps for phase estimation and provides an analysis of errors and the success probability.

First, the illustrative example of the 
single-qubit phase estimation is presented, then we review the multi-qubit setting. We follow closely references \cite{Cleve1997,nielsen2002quantum,Xu2018}.
The single-qubit phase estimation is for demonstration purposes since it uses a known phase $\theta$. 
Here, we apply the single-qubit gate
\begin{equation}
\mathcal U_z=e^{-i \frac{\theta}{2} \sigma_z} 
\end{equation}
for varying powers. 
In this treatment, the single qubit is effectively simulating an $m$-qubit register.
To obtain the least significant bit of the phase, the first step is to apply $\mathcal U_z^{M/2}$, where $M$ is such that $M \geq 2 \pi / \epsilon$ and $M = 2^{m}$ for an integer $m$:
\begin{eqnarray}
\frac{ \left( \ket 0 + \ket 1 \right)}{\sqrt 2} &\to& \frac{1}{\sqrt 2} \left( e^{-i \frac{\theta}{2} \frac{M}{2}} \ket 0 +  e^{+i \frac{\theta}{2} \frac{M}{2}}  \ket 1 \right)  \\
 &\to& \frac{1}{2} \left( \left(e^{-i \frac{\theta}{2} \frac{M}{2}} +e^{+i \frac{\theta}{2} \frac{M}{2}} \right) \ket 0 + \right . \nonumber \\
 && \left .  \left (e^{-i \frac{\theta}{2} \frac{M}{2}} - e^{+i \frac{\theta}{2} \frac{M}{2}} \right ) \ket 1 \right) \\
  &\to& \frac{1}{2} \left( 2 \cos \left (\frac{\theta M}{4} \right ) \ket 0 +  2 i \sin \left (\frac{\theta M}{4} \right)  \ket 1 \right). \nonumber
\end{eqnarray}
The measurement probabilities are
\begin{equation}\label{eqProbm}
P_0^{(m)} = \cos^2 \left (\frac{\theta M}{4} \right ), \quad P_1^{(m)} = \sin^2 \left (\frac{\theta M}{4} \right ).
\end{equation}
We now use $\mathcal U_z^{k/2}$  for $k=m-1,\dots,1$ to estimate the remaining bits of the phase,
\begin{eqnarray}
&&\frac{1}{\sqrt 2}\left( \ket 0 + \ket 1 \right) \to \frac{1}{\sqrt 2} \left( e^{-i \frac{\theta}{2}  2^{k-1}} \ket 0 +  e^{+i \frac{\theta}{2}  2^{k-1}}  \ket 1 \right)   \nonumber \\
&\to& \frac{1}{\sqrt 2} \left( e^{-i \frac{\theta}{2}  2^{k-1}} \ket 0 +  e^{-i \frac{\theta'}{2} 2^{k-1} } e^{+i \frac{\theta}{2}  2^{k-1}}  \ket 1 \right).
\end{eqnarray}
The last step applies a phase $e^{-i \theta' 2^{k-1}}$ to $\ket 1$ given by the known bits, using the identity 
\begin{equation}
e^{-i \theta' 2^{k-1}} = e^{-i \pi [.b_{k+1},\dots,b_m] }.
\end{equation}
After another Hadamard gate, we obtain
\begin{eqnarray}
\to \frac{1}{\sqrt 2} \left( ( e^{-i \frac{\theta}{2}  2^{k-1}} + e^{-i \frac{\theta'}{2} 2^{k-1} } e^{+i \frac{\theta}{2}  2^{k-1}} )\ket 0 \right. \nonumber \\ + \left. ( e^{-i \frac{\theta}{2}  2^{k-1}}  - e^{-i \frac{\theta'}{2} 2^{k-1} } e^{+i \frac{\theta}{2}  2^{k-1}} ) \ket 1 \right).
\end{eqnarray}
From this, the measurement probabilities are given by \cite{Xu2018}
\begin{eqnarray}\label{eqProbk}
P_0^{(k)} &=& \frac{1}{2} + \frac{1}{2}\cos \left (2^k \theta   - \pi [.b_{k+1},\dots,b_m] \right ), \nonumber \\ 
P_1^{(k)} &=&  1- P_0^{(k)}.
\end{eqnarray}
These probabilities can be sampled to obtain an $m$-bit estimate of $\theta$.

In the main text and Fig.~\ref{figure1}, we obtain the output probability distribution for the best $m$-bit estimate for the phase $\theta$ using standard $m$-qubit phase estimation, as typically discussed in the literature \cite{Cleve1997,nielsen2002quantum}. This presents a coherent, controlled version of the procedure above. Using the eigenstate $ \ket {\psi_\theta}$ associated with the eigenvalue $\theta$, an $m$-qubit register, and the operation $\mathcal Q^c$, Eq.~(\ref{eqQc}), we can perform 
\begin{equation}
\sum_{y=0}^{2^m-1} \ket y \mathcal Q^y \ket {\psi_\theta}.
\end{equation}
In the $m$-qubit register, we obtain
\begin{eqnarray}
(\ket 0 + e^{i2 \pi 2^{m-1} \theta} \ket 1)(\ket 0 + e^{i2 \pi 2^{m-2} \theta} \ket 1) \times \\  \times \dots \times (\ket 0 + e^{i2 \pi \theta} \ket 1)\nonumber \\ = \sum_{y=0}^{2^m -1} e^{i2 \pi \theta y} \ket y. \nonumber
\end{eqnarray}
After applying the inverse Quantum Fourier transform we have the state 
\begin{equation}
\frac{1}{2^m} \sum_{x=0}^{2^m -1} \sum_{y=0}^{2^m -1} e^{-i2\pi \frac{x y}{2^m}} e^{i2 \pi \theta y} \ket x.
\end{equation}
Let $\hat \theta$ be the  $m$-bit approximation of $\theta$ and $\theta = \hat \theta +\delta$. Using these definitions leads to
\begin{eqnarray}
\frac{1}{2^m} \sum_{x=0}^{2^m -1} \sum_{y=0}^{2^m -1} e^{-i2\pi \frac{x y}{2^m}} e^{i2 \pi (\hat \theta +\delta) y} \ket x \nonumber
= \\ \frac{1}{2^m} \sum_{x=0}^{2^m -1} \sum_{y=0}^{2^m -1} e^{i2\pi \frac{(2^m \hat \theta - x)y}{2^m}} e^{i2 \pi\delta y} \ket x \\ \label{eqPhaseEstState}
=:\sum_{x=0}^{2^m -1} \alpha_\theta(x) \ket x.
\end{eqnarray}
The estimate $\hat \theta$ 
has the amplitude
\begin{equation}
\alpha_\theta(2^m \hat \theta)  = \frac{1}{2^m} \sum_{y=0}^{2^m -1} e^{i2 \pi\delta y} = \frac{1}{2^m} \left( \frac{1-e^{i2\pi \delta2^m}}{1-e^{i2\pi \delta} }\right),
\end{equation}
using the geometric series. This occurs with probability
\begin{equation}\label{eqSuccessBest}
P(\hat \theta) =  \frac{1}{4^m} \left \vert \frac{1-e^{i2\pi \delta 2^m}}{1-e^{i2\pi \delta} } \right \vert^2 
=  \frac{1}{4^m} \left \vert \frac{1-e^{i2\pi \theta 2^m}}{1-e^{i2\pi (\theta -\hat \theta)} } \right \vert^2.
\end{equation}
One can efficiently sample from this bit-string distribution via the individual bit probabilities 
given in Eqs.~(\ref{eqProbm}) and (\ref{eqProbk}).

We now provide an error analysis of phase estimation. We follow closely the discussion in previous references. 
The phase estimation algorithm provides an estimate $\hat \theta$ that is accurate with $\epsilon$, thus we have
\begin{equation}
\vert \hat \theta -  \theta \vert \leq \epsilon.
\end{equation}
The quantity of interest is the expectation value, which is related to the phase $\theta$ from Eq.~(\ref{eqMeanToAngle}) as
\begin{equation}
1-2 \mu = \cos \frac{\theta}{2}.
\end{equation}
We would like to determine a bound for the accuracy of this expectation value, 
i.e. determine 
\begin{equation}
\vert \hat \mu - \mu \vert.
\end{equation}
First, we can use the Taylor expansion $\cos ((\theta \pm \epsilon)/2) = \cos\theta/2 - (\pm \epsilon) \sin (\theta/2) + \Ord{\epsilon^2}$ to arrive at the first-order bound
\begin{equation}
\vert \hat \mu - \mu \vert \leq \Ord{\frac{\epsilon}{2} \sin{\frac{\hat \theta}{2}}}.
\end{equation}
More generally, we can show the following.
\begin{lemma}
Assume $\vert \hat \theta -  \theta \vert \leq \epsilon$, $0 \leq \hat \theta < \pi$, and $0 < \epsilon \leq 1$. Then 
\begin{equation}
\vert \hat \mu - \mu \vert \leq \vert \cos ((\hat \theta + \epsilon)/2) - \cos \hat \theta /2 \vert.
\end{equation}
\end{lemma}
\begin{proof}
Use the trigonometric identity for the difference between cosines 
\begin{eqnarray}
&\vert \cos \hat \theta/2 - \cos  \theta/2 \vert = 2 \vert \sin ((  \hat\theta+  \theta)/4)\sin(( \hat \theta-  \theta)/4)\vert, \nonumber \\
&\vert \cos ((\hat \theta + \epsilon)/2) - \cos \hat \theta /2 \vert=
 2 \vert \sin ((  2\hat\theta+  \epsilon)/4)\sin(\epsilon/4)\vert. \nonumber
\end{eqnarray}
Note that by definition the bit estimate $\hat \theta \leq \pi - \frac{\epsilon}{2}$, thus, with $\vert \hat \theta -  \theta \vert \leq \epsilon$, we have  $\frac{ \hat\theta+  \theta}{4}\leq \frac{2\hat \theta +\epsilon}{4} \leq \frac{\pi} 
{2}$. Thus
\begin{equation}
\left| \sin \left (\frac{\hat \theta +\theta}{4} \right ) \right| \leq  \left| \sin \left (\frac{2 \hat \theta +\epsilon}{4} \right) \right|.
\end{equation}
Also 
\begin{equation}
\left| \sin\left( \frac{\hat \theta-  \theta}{4} \right )\right| \leq \left| \sin \left (\frac{\epsilon}{4} \right) \right|.
\end{equation}
\end{proof}

We now discuss increasing the probability of success for phase estimation. The probability of observing the best $m$-bit approximation is lower bounded by $8/\pi^2> 0.81$, from Eq.~(\ref{eqSuccessBest}) \cite{nielsen2002quantum}. To boost this success probability, multiple runs of phase estimation can be performed. The median of these multiple runs will have a higher success probability \cite{Nagaj2009,Montanaro2015,Xu2018}, as will be shown now.  
Let $\hat \theta_1, \dots, \hat \theta_D$ be the results of $D$ independent runs of phase estimation. 
The new estimate is the median $\hat \theta = {\rm Median} (\hat \theta_1, \dots, \hat \theta_D)$.
\begin{lemma}[\cite{Nagaj2009}] \label{lemmaMedian}
Let the desired accuracy be $\epsilon>0$. Let the probability that each sample falls outside the accuracy, i.e. $\vert \hat \theta_j - \theta\vert \geq \epsilon$,  be $0<\delta<1/2$. Then the probability that
the median is inaccurate is bounded by $p_f\leq \frac{1}{2} \left( 2 \sqrt{\delta(1-\delta)} \right)^D$.
\end{lemma}
The proof is provided in~\cite{Nagaj2009}.
The confidence/success probability is defined as $c := 1- p_f$. Taking the logarithm of the failure probability, 
$\log 2(1-c) \leq D \log 2 \sqrt{\delta(1-\delta)}$, leads to 
\begin{equation}
\vert \log 1-c \vert \geq  D \vert \log 2 \sqrt{\delta(1-\delta)} \vert,
\end{equation}
which leads to 
\begin{equation}
D = \Ord{\vert \log 1-c \vert },
\end{equation}
if $0<\delta<1/2$ is a constant. Hence, for a confidence of $c$ one needs at most $\Ord{\vert \log 1-c 
\vert }$ independent repetitions of phase estimation. 

\end{document}